\documentclass[11pt,a4paper]{article}

\usepackage{amsmath, amscd, amssymb, amsthm}
\usepackage{xcolor,graphicx}
\usepackage{hyperref}
\usepackage{url}
\usepackage[all]{xy}
\usepackage{ulem}

\allowdisplaybreaks
\newtheorem{theorem}{Theorem}[section]
\newtheorem*{theorem*}{Theorem}
\newtheorem{proposition}[theorem]{Proposition}
\newtheorem{lemma}[theorem]{Lemma}
\newtheorem{dlemma}[theorem]{Definition-Lemma}
\newtheorem{fact}[theorem]{Fact}
\newtheorem{corollary}[theorem]{Corollary}
\theoremstyle{definition} 
\newtheorem{definition}[theorem]{Definition}
\newtheorem{remark}[theorem]{Remark}
\newtheorem{example}[theorem]{Example}
\newtheorem{notation}[theorem]{Notation}

\newcommand\+[1]{\mathcal{#1}}
\newcommand{\N}{{\mathbb N}}     
\newcommand{\PP}{{\mathbb P}}    
\newcommand{\Z}{{\mathbb Z}}     
\newcommand{\Q}{{\mathbb Q}}     
\newcommand{\RR}{{\mathbb R}}    

\newcommand{\kernel}{\textit{Ker}}

\newcommand{\suc}{\textit{Suc}}

\newcommand{\lcm}{\textit{lcm}}

\newcommand{\Zhat}{\widehat{\Z}}    
\newcommand{\val}{\textit{Val}}     
\newcommand{\DUO}{\textit{DUO}}     
\newcommand{\gen}{\textit{gen}}   
\newcommand{\Suc}{\textit{Suc}}

\newcommand{\seg}{\color{blue}}     
     
\newcommand{\sege}{\color{cyan}} 
\definecolor{olive}{rgb}{0.5, 0.5, 0.0}
    \newcommand{\ire}{\color{olive}} 

{\end{list}}%

{\end{list}}%

{\end{list}}%

{\end{list}}%

\begin{document}
\title {Congruence Preservation, Lattices and Recognizability}
\author{
{Patrick C\'egielski\textsuperscript{1,3}}
\and 
{Serge Grigorieff\textsuperscript{2,3} }
\and 
{Ir\`ene  Guessarian\textsuperscript{2,3,4}}
}
\maketitle

\footnotetext[1]{LACL, EA 4219, Universit\'e Paris-Est Cr\'eteil,  IUT S\'enart-Fontainebleau, France
\texttt{cegielski@u-pec.fr}}
\footnotetext[2]{IRIF,  UMR 8243, Universit\'e Paris 7 Denis Diderot, France \texttt {FirstName.LastName@irif.fr}}
\footnotetext[3]{Partially supported by TARMAC ANR agreement 12 BS02 007 01}
\footnotetext[4]{Emeritus at UPMC Universit\'e Paris 6}
\bibliographystyle{plain}

\maketitle
\begin{center}
{\sl To Yuri, on his 70th birthday, with   gratitude for sharing with us his beautiful mathematical ideas during our long friendship}

\end{center}

\begin{abstract} Looking at some monoids and (semi)rings (natural numbers, integers and $p$-adic integers),
and more generally, residually finite algebras (in a strong sense),
we prove  the equivalence of two ways for a function on such an algebra to behave like the operations of the algebra. The
first way is to preserve congruences or stable preorders.
The second way is to demand that preimages of recognizable sets belong to the lattice 
 or the Boolean algebra generated by the
preimages of recognizable sets by ``derived unary operations'' of the algebra (such as translations, quotients,\ldots). 
\end{abstract}
{\scriptsize\tableofcontents}

\section{Motivation and overview of the paper } 
\label{s:intro}
%
%
In \cite{cgg14ipl}, 
we proved that
if $f:\N\longrightarrow\N$ is non decreasing then conditions (1) and (2) are equivalent
\begin{description}
\item[(1)] 
 \begin{description}
\item[(a)]  for all $a,b\in\N$, $a-b$ divides $f(a)-f(b)$, and 
\item[(b)] for all $a\in\N$, $f(a)\geq a$,
\end{description}
\item[(2)]  every lattice  $\+L$ of  regular  subsets of $\N$ which is closed under $x\mapsto x-1$, i.e., $L\in\+L$ implies $\{n\mid n+1\in L\}\in\+L$, is  also closed under $f^{-1}$,  i.e., for every $L\in\+L$, $f^{-1}(L)=\{n\in\N\mid f(n)\in L\}\in \+L$.
\end{description}

For which (semi)rings does this property   hold\,?    For instance, does it hold for  the ring of integers $\Z$ or  the rings of $p$-adic integers $\Z_p$\,? To extend this property to diverse  structures, we begin by rewriting the two conditions (1) and (2)
in more algebraic terms.

Observing that condition (1) is equivalent to  the notion of  ``congruence preservation" (Section \ref{ss:CPD}
Theorem \ref{thm:CPsurN}) in  the case of  $\N$,  we  will use the latter notion of congruence preservation instead of condition (1). The general notion of congruence  
preservation is defined in Definition \ref{def:congCompat 1} for arbitrary algebras.
This will allow to consider general algebras in the sense of universal algebra instead of just (semi)rings.

Moreover, as regular subsets coincide with recognizable subsets for $\N$ (Remark \ref{rk:rr}), we will  use ``recognizable" subsets  in condition {\bf (2)} (see Section \ref{sec:rec}) instead of regular subsets, again leading to an algebraic statement also suitable for general algebras. 

The above equivalence can thus be restated for the original case $\langle \N ;  +, \times\rangle$ as 
\begin{theorem}\label{thm:ipl1}
If $f:\N\longrightarrow\N$ is non decreasing 
 then conditions (1) and (2) below are equivalent
\begin{description}
\item[(1)] $f$ is  congruence preserving on $\langle \N ;  +, \times\rangle$ and, for all $a\in\N$, $f(a)\geq a$
\item[(2)] for every recognizable subset $L$ of \ $\langle \N; +, \times\rangle$ the smallest  lattice of subsets of $\N$ containing $L$  and closed under $x\mapsto x-1$ is  also closed under $f^{-1}$.
\end{description}
\end{theorem}
In the present paper, we try and generalize Theorem \ref{thm:ipl1} as much as possible: i.e., for which classes of algebras does a similar Theorem hold\,?  We  investigate in a general framework the relationships  between congruence preservation, recognizability and lattices 
 or Boolean algebras of recognizable sets.

\medskip
Formal definitions are recalled in Section \ref{s:def}.  Besides the usual universal algebra notion of congruence preservation, we consider a similar notion of stable preorder preservation.  We also extend to general algebras the notions of recognizability, syntactic preorder and syntactic congruence from language theory.

In Section \ref{ss:frying pan}, we prove Theorem \ref{thm:ipl1} (cf. Theorem \ref{thm:ipl} in Section \ref{s:35}).  To this end, using the characterization of congruences on $\langle\N;+\rangle$ (Sections \ref{s:CN+} and \ref{s:FryingPan}) we prove that congruence preserving functions on  $\langle\N;+\rangle$ are exactly those satisfying conditions {\bf (1) (a)} and  {\bf (b)} supra, a result interesting per se (Theorem \ref{thm:CPsurN}, Section \ref{ss:CPD}).
Functions satisfying condition {\bf (1) (a)} on $\langle\N;+\rangle$  have been characterized in \cite{cgg14}.  They
can be very complex, for instance $x\mapsto$ if $x=0$ then 1 else $\lfloor e x!\rfloor$ satisfies conditions {\bf (1) (a)}  and  {\bf (b)}.
In Section \ref{s:34} we prove that stable preorder preserving functions on $\langle\N;+\rangle$ are exactly the non decreasing congruence preserving functions (Theorem 
 \ref{prop:miracle}). All these results on $\langle\N;+\rangle$  hold on $\langle\N;+,\times\rangle$. In Section   \ref{ss:Ntimes}, we generalize Theorem  \ref{thm:ipl1}
to the monoid $\langle\N;\times\rangle$.
Moreover, we give a very simple characterization of the corresponding subclass of congruence preserving functions: this subclass consists of all monomial functions $x\mapsto kx^n$ (Theorem \ref{thm:iplX}).

We prove our main results  in Section \ref{s:genResFin}. We consider variants of Theorem 1.1 for algebras as general as possible
and a version of condition {\bf (2)} of Theorem 1.1
involving the lattices $\+L_\+A(L)$ 
and the Boolean algebras $\+B_\+A(L)$ of preimages of a recognizable set $L$
by derived unary operations of the algebra (such as translations, quotients,\ldots, cf. Definition~\ref{def:Lat-LAL}).
Our results show that congruence preservation of a function $f$ 
is related to the condition that $f^{-1}(L)$ belongs to the Boolean algebra $\+B_\+A(L)$ for all recognizable $L$
whereas stable preorder preservation  
is related to the condition that $f^{-1}(L)$ belongs to the lattice $\+L_\+A(L)$.
 Theorem~\ref{th:LatticeGen} is a general wild version of Theorem 1.1
relating stable preorder preserving functions to the
condition $f^{-1}(L)$ belongs to the complete lattice variant of $\+L_\+A(L)$ for all sets $L$.
Theorem~\ref{th:RecLatGeneralResiduallyFinite} in section~\ref{s:43}
shows that, on any algebra,
stable preorder preserving functions satisfy $f^{-1}(L)\in\+L_\+A(L)$ for recognizable $L$.
The reciprocal is true for  sp-residually  finite algebras
(a strong variant of residual finiteness, cf. Definition~\ref{def:residually finite c sp}).
 Avatars with congruence preservation and the Boolean algebra $\+B_\+A(L)$
are stated in Theorems~\ref{thm:cp bhull} and \ref{th:RecLatGeneralResiduallyFinite c}.
 In case the algebra contains a group operation and satisfies a strong form of residual finiteness
it turns out that all the conditions considered in the paper are equivalent.

Section \ref{ss:Z} is devoted to $\langle\Z;+,\times\rangle$  to which
Theorem~\ref{th:RecLatGeneralResiduallyFinite} applies, giving Theorem~\ref{thm:mainZ1} in Section \ref{s:55}.
Though the congruence preserving functions can be very intricate \cite{cgg14B},
for instance,
$$
f(n)= \left\{\begin{array}{ll}
\sqrt{\dfrac{e}{\pi}}
\times  \dfrac{\Gamma(1/2)}{2\times4^n\times n!}
\displaystyle\int_1^\infty e^{-t/2}(t^2-1)^n dt
&\quad\text{for $n\geq 0$}
\\
-f(|n|-1)&\quad\text{for $n<0$}
\end{array}\right.\,,
$$
the lattices $\+L_{\langle\Z;+,\times\rangle}(L)$ for $L$ recognizable are very simple, cf. Lemma~\ref{l:lattice L Z} in Section 5.4.
 
In Section \ref{ss:ZpZhat}, congruence preservation for the rings of $p$-adic integers
is treated similarly  (Proposition \ref{prop:CPlatZp}, section \ref{s:latticesZp}). 
\section{Preliminary definitions }\label{s:def}

We here recall the useful definitions, notations and prove basic results.
\subsection{Stable relations  and congruences on an algebra}
\begin{definition} An {\em algebra } $\+A=\langle A;\Xi\rangle$ consists of a nonempty carrier set $A$ together with a set of operations $\Xi$, each $\xi\in\Xi$ is  a mapping $\xi\colon A^{ar(\xi)}\to A$ where $ar(\xi)\in\N$ is the arity of $\xi$.
\end{definition}
\begin{definition}\label{def:fpreserveR}
Let $A$ be a set and let $f$ be a function $f\colon A^p\longrightarrow A$. A binary relation $\rho$ on  $A$ 
is said to be {\em compatible} with $f$ if and only if, for all elements
$x_1,\ldots,x_{p}, y_1,\ldots,y_{p}$ in $A$
\begin{equation}\label{eq:preserveR}
(x_1\rho\, y_1\ \wedge\cdots\wedge\ x_{p}\rho\, y_{p})\quad
  \Longrightarrow \quad f(x_1,\ldots,x_p)\;\rho\, f(y_1,\ldots,y_p)
\end{equation}  
\end{definition}
\begin{definition}\label{def:stable}  A  binary relation $\rho$   on  $A$ is said to be {\em stable} on the algebra $\+A=\langle A;\Xi\rangle$  
 if it is compatible with each  operation $\xi\in\Xi$,
 i.e., if $\xi$ is $n$-ary then,  for all $x_1,\ldots,x_{n},\allowbreak y_1,\ldots,y_{n}$ in $A$
\begin{equation}\label{eq:stable}
(x_1\rho \; y_1\ \wedge\cdots\wedge\ x_{n}\rho  \; y_{n})
\quad \Longrightarrow\quad 
 \xi(x_1,\ldots,x_{n})\;\rho \; \xi(y_1,\ldots,y_{n})
\end{equation}
\end{definition}
\begin{definition}
A stable  equivalence relation  on $\+A$ is called an  {\em $\+A$-congruence}. If there are finitely many equivalence classes, it is said to be a {\em finite index} congruence.
\end{definition}
%
%
\subsection{Congruence  and stable (pre)order preservation}
\subsubsection{Definitions}
The {\em substitution property}, introduced by Gr\"atzer  in \cite{gratzer} page 44,   has since been renamed {\em congruence preservation} in the literature.  We shall also use an extension dealing with (pre)orders instead of congruences.
\begin{definition}\label{def:congCompat 1}
Let $\+A=\langle A;\Xi\rangle$ be an algebra.

1) A function $f\colon A^p\longrightarrow A$
is {\em $\+A$-congruence preserving} if  all $\+A$-congruences are compatible with $f$, i.e., for every congruence $\sim$ on  $\+A$ and all elements
$x_1,\ldots,x_{p}, \allowbreak  y_1,\ldots,y_{p}$ in $A$
\begin{equation}\label{eq:cp}
(x_1\sim y_1\ \wedge\cdots\wedge\ x_{p}\sim y_{p})\quad
  \Longrightarrow \quad f(x_1,\ldots,x_p)\sim f(y_1,\ldots,y_p).
\end{equation}  

2) A function $f\colon A^p\longrightarrow A$ is {\em $\+A$-stable 
 (pre)order preserving}
if all $\+A$-stable  (pre)orders are compatible with $f$, i.e., for every stable   (pre)order $\preceq$ on  $\+A$ and all elements
$x_1,\ldots,x_{p}, y_1,\ldots,y_{p}$ in $A$
\begin{equation}\label{eq:spp}
(x_1\preceq y_1\ \wedge\cdots\wedge\ x_{p}\preceq y_{p})\quad
  \Longrightarrow \quad f(x_1,\ldots,x_p)\preceq f(y_1,\ldots,y_p).
\end{equation}
\end{definition}
When the algebra $\+A$ is clear from the context,  $f$ is simply said to be {\em congruence preserving} (resp. {\em stable (pre)order preserving}).

Congruences and congruence  preservation can also be defined in terms of morphisms.

\begin{definition}
For $f\colon A\rightarrow B$, the kernel $Ker(f)$ of  $f$ is defined by $Ker(f)=\{(x,y)\; | \; f(x)=f(y)\}$.
\end{definition}

\begin{lemma} \label{CongVsHomo}  1) A binary relation on $\+A$ is a congruence if and only if it is the {\em kernel}
$\kernel(\varphi)=\{(x,y)\mid\varphi(x)=\varphi(y)\}$ of some homomorphism 
$\varphi\colon A\rightarrow B$ onto some algebra $\+B$.

2) $f\colon A\rightarrow A$ is congruence preserving if and only if, for every homomorphism $\varphi\colon A\rightarrow B$, $Ker(\varphi)\subseteq Ker(\varphi\circ f)$.
\end{lemma}
The next result shows that congruence preserving functions  somehow extend  operations of the algebra.
\begin{proposition}\label{thm:CP=FactorThroughMorphism}
Let $\+A=\langle A;\Xi\rangle$ be an algebra
and let $f:A^n\to A$ with $n\geq1$.
The following conditions are equivalent:
\begin{enumerate}
\item[(i)]
$f$ is $\+A$-congruence preserving,
\item[(ii)]
For every algebra $\+B=\langle B;\Theta\rangle$ having the same signature as $\+A$ and every surjective morphism $\varphi\colon A \to B$ there exists a unique 
function $f_\varphi\colon B^n\to B$ such that 
$\varphi(f(x_1,\ldots,x_n)) = f_\varphi(\varphi(x_1),\ldots,\varphi(x_n))$,
(i.e., $\varphi$ is also a morphism between the algebras $\langle A;\Xi\cup\{f\}\rangle$ and
$\langle B;\Theta\cup\{f_\varphi\}\rangle$,
namely the  diagram of Figure~\ref{NEW fig:fak} is commutative). 
\end{enumerate}
\begin{figure}[h]
\[\xymatrix{
A^n \ar[rr]^{\text{\normalsize$f$}} 
 \ar[d]_{\text{\normalsize$(\varphi,\ldots,\varphi)$}}&& A\ar[d]^{\text{\normalsize$\varphi$}}
\\
B^n  \ar[rr]_{\text{\normalsize$f_\varphi$}}  && B
}\]
\caption{From $f\colon A^n\to A$ to $f_\varphi \colon B^n \to B$}\label{NEW fig:fak}
\end{figure}
\end{proposition}

\begin{proof} $(i)\Longrightarrow (ii)$ Assume $(i)$. 
As $\varphi$ is a morphism, $Ker(\varphi)$ is the congruence:
$x\sim y$ if and only if $\varphi(x)=\varphi(y)$;
 as $f$ is congruence preserving, $\varphi(x_i)=\varphi(y_i)$ for $1\leq i\leq n$ implies $\varphi(f(x_1,\ldots,x_n))=\varphi(f(y_1,\ldots,y_n))$  hence $f_\varphi$  is well defined and equal to the common value of all the $ \varphi(f(y_1,\ldots,y_n))$  for $y_i$'s such that $ \varphi(y_i)= \varphi(x_i)$. 

$(ii)\Longrightarrow (i)$  Assume (ii). Let $\sim$ be a congruence on $A$ and let $\+B=\langle A/\!\!\sim;\Xi/\!\!\sim\rangle$ be the quotient algebra. The canonical quotient map $\varphi\colon A\to A/\!\!\sim$ is a surjective morphism.
By $(ii)$, $f$ factors through $A/\!\!\sim$ to $f_\varphi$.
In particular, if $y_i\sim  x_i$ for $1\leq  i \leq n$ then
$\varphi(y_i)=\varphi(x_i)$ hence 
$f_\varphi(\varphi(y_1),\ldots,\varphi(y_n))=f_\varphi(\varphi(x_1),\ldots,\varphi(x_n))$.
Using $(ii)$, we get
$\varphi(f(y_1,\ldots,y_n)) = \varphi(f(x_1,\ldots,x_n))$,
i.e.,  $f(y_1,\ldots,y_n) \sim f(x_1,\ldots,x_n)$.
Hence $f$ preserves congruences.
\end{proof}

\subsubsection{Reduction from arity $n$ to arity one}
Congruence preservation of a function of arbitrary arity can be characterized via congruence preservation of its restrictions to unary functions. This enables us to simplify some proofs.  

Reducing to
unary functions  is also a key point in the definition of recognizability, syntactic congruences and syntactic preorders
for general algebras.
\begin{definition}\label{not:freeze}
Given $n\geq2$, a $n$-ary $f:A^n\to A$,  an index $i\in\{1,\ldots,n\}$, and 
$\vec{c}=(c_1,\ldots,c_{i-1},c_{i+1},\ldots,c_n)\in A^{n-1}$, we denote by
$f_i^{\vec{c}}$ the unary function $A\to A$ (called the {\em frozen function } of $f$ relative to $i$, $\vec c$) obtained by fixing all arguments to $\vec c$ except the $i$-th one. In other words,
$$
f_i^{\vec{c}}(x)=f(c_1,\ldots,c_{i-1},x,c_{i+1},\ldots,c_{n}).
$$
\end{definition}

\begin{lemma}\label{lem:de n a 1}
1) An equivalence relation (resp.  (pre)order) is compatible with an $n$-ary function $f\colon A^n\longrightarrow A$ if and only if for all $i\in\{1,\ldots,n\}$, for all  $\vec{c} =(c_1,\ldots,c_{i-1},c_{i+1},\ldots,c_n)\in A^{n-1}$, it is compatible with the unary function $f^{\vec{c}}_i$.

2) Given an algebra $\+A=\langle A;\Xi\rangle$, a $n$-ary function $f\colon A^n\longrightarrow A$ 
 is $\+A$-congruence preserving (resp. $\+A$-stable (pre)order preserving) if and only if for all $i\in\{1,\ldots,n\}$, for all  $\vec{c} =(c_1,\ldots,c_{i-1},c_{i+1},\ldots,c_n)\in A^{n-1}$,  the unary function
$f^{\vec{c}}_i$ is $\+A$-congruence preserving (resp. $\+A$-stable (pre)order preserving).
\end{lemma}
\begin{proof} We prove 1) for an equivalence $\sim$, the case of preorders is similar, and 2) is an immediate consequence of 1). The left to right implication in 1) is clear. For the converse implication, use the transitivity of $\sim$, e.g.,  assuming $n=2$, if $\sim$ is compatible with $f^{a_1}_1$ and $f^{b_2}_2$, then $a_1\sim b_1$ and $a_2\sim b_2$  imply $f(a_1,a_2)\sim f(a_1,b_2)$ and $f(a_1,b_2)\sim f(b_1,b_2)$, hence $f(a_1,a_2)\sim  f(b_1,b_2)$.
\end{proof}
%
\subsubsection{Syntactic congruence and preorder}

To every  subset of the carrier set of an algebra are  associated a 
 {\em syntactic congruence} and a   {\em syntactic preorder}. 
Let us  first define the notion of {\em derived unary operation.}

\begin{definition} \label{def:gen} [Derived Unary Operations] Given $\+A=\langle A;\Xi\rangle$, 
we denote by $ {\DUO}(\+A)$  the set of unary functions $\gamma$ defined by composing frozen functions of the operations in $\Xi$, i.e., $\gamma=\xi_{1,j_1}^{\vec {c_1}}\circ \cdots\circ\xi_{n,j_n}^{\vec {c_n}}$ where
  $n\in\N\setminus\{0\}$,  $i=1,\ldots,n,\ \xi_i\in\Xi$, $1\leq j_i\leq ar(\xi_i)$, and $\vec{c_i}\in A^{ar(\xi_i)-1}$. For $n=0$, $\gamma$ is the identity on $A$. 
\end{definition}
\begin{example} 1) For  $\+A=\langle \N;+\rangle$,  ${\textit{DUO}}(\+A)$ is the set of translations $x\mapsto x+a$, for $a\in\N$. 

2) For  $\+A=\langle \N\setminus\{0\};\times\rangle$
(resp. $\+A=\langle \N;\times\rangle$), 
 ${\textit{DUO}}(\+A)$ is the set of homotheties $x\mapsto ax$, 
with $a\geq1$ (resp. $a\in\N$).

3)  For $\+S=\langle\Sigma^*; \cdot \rangle$, the algebra of words with concatenation,  ${\textit{DUO}}(\+A)$ is the set of left and right multiplications by words $x\mapsto w\cdot x\cdot w'$ for $w,w'\in\Sigma^*$.
\end{example}
%
%
%
Using the notion of Derived Unary Operations, we can define the syntactic preorder and 
syntactic congruence associated with $L$.   

 \begin{dlemma} \label{syntacticL_congruence}  For $L\subseteq A$, 
the relation $\leq_L $  defined by  
\begin{equation}\label{eq:sord}
x\leq_L  y\quad {\text{ if and only if } } \quad\forall \gamma\in {\textit{DUO}}(\+A) \quad \big(\gamma(y)\in L \Longrightarrow \gamma(x)\in L \big)
\end{equation} 
   is a stable preorder. It is called the {\em syntactic preorder} associated with $L$.
   
The relation $\sim_L $  defined by 
\begin{equation}\label{eq:scong}
 x\sim_L  y \quad {\text{ if and only if } } \quad\forall \gamma\in {\textit{DUO}}(\+A) \quad \big(\gamma(x)\in L \Longleftrightarrow \gamma(y)\in L \big)
\end{equation} 
   is the congruence associated with  the preorder $\leq_L $. It is called the {\em syntactic congruence} associated with $L$. 
 \end{dlemma}       
\begin{proof}  It is clear that $\leq_L $ is  reflexive and transitive.
Also,  if $x\leq_L  y$ and $\delta\in {\textit{DUO}}(\+A)$ then equation \eqref{eq:sord}
(applied with the composition $\gamma\circ\delta$) insures that
for all $\gamma\in {\textit{DUO}}(\+A)$ we have 
$\gamma(\delta(x))\in L\iff \gamma(\delta(x))\in L$ ; hence
$\delta(x)\leq_L  \delta(y)$.
Applying Lemma \ref{lem:de n a 1}, we see that $\leq_L $ is a stable preorder.

 Clearly, $\sim_L $ is reflexive, symmetric and transitive,  and  it is the equivalence associated with $\leq_L $. 
By equation \eqref{eq:scong}, we see that for all $\xi\in\Xi$ , for $1\leq i\leq ar(\xi)$ and $\vec c\in A^{ar(\xi)-1}$,  $\xi_i^{\vec c}$  is compatible with $\sim_L $, hence Lemma \ref{lem:de n a 1} implies that for all $\xi\in\Xi$, $\xi$  is compatible with $\sim_L $, i.e., $\sim_L $ is a congruence on $\langle A;\Xi\rangle$.
\end{proof}
\begin{remark}
Recall, that in the algebra of words $\Sigma^*$ with concatenation, 
(1)  the frozen unary operations consist in adding a fixed prefix or suffix,
(2) the family $\textit{DUO}$ consists of operations $x\mapsto uxv$
for fixed $u,v\in\Sigma^*$. if $L\subseteq\Sigma^*$ is a language then 
its syntactic congruence $x\sim_L  y$ is defined by the condition
$\forall u,v\in\Sigma^*\ (uxv\in L\Leftrightarrow uyv\in L)$.
Our notion of syntactic congruence thus generalizes the usual notion of syntactic congruence in language theory.
\end{remark}
\begin{definition}\label{df:satureCong} A set is said to be {\em  saturated with respect to an equivalence} if is is a union of equivalence classes. 
\end {definition}

Proposition \ref{p:syntactic largest} states some properties of syntactic congruences and preorders. 

\begin{proposition}\label{p:syntactic largest}
Let $L$ be a subset of an algebra $\+A$.

1) If $L$ is saturated for a congruence $\equiv$ of $\+A$ 
then $\equiv$ refines the syntactic congruence $\sim_L $ of $L$, i.e.,
$x\equiv y$ implies $x\sim_L  y$.

2) If $L$ is an initial segment of a stable preorder $\preceq$ of $\+A$ (i.e.,  if  $b\in L$ and  $x\preceq b$ then $x\in L$),
then $\preceq$ refines the syntactic congruence $\leq_L $ of $L$, i.e.,
$x\preceq y$ implies $x\leq_L  y$.
\end{proposition}
\begin{proof}
Assume $x\equiv y$. Since $\equiv$ is a congruence we have $\gamma(x)\equiv \gamma(y)$
for all $\gamma\in\DUO(\+A)$. If $L$ is saturated for $\equiv$ we then have
$\gamma(x)\in L \Leftrightarrow \gamma(y)\in L$ hence $x\sim_L  y$.
Similar proof with a stable preorder.
\end{proof}

\subsection{Recognizability}\label{sec:rec}

%
\begin{definition} [Recognizability] \label{def:rec}
Given $\+A=\langle A; \Xi\rangle$ an algebra, 
a  subset $B$ of $A$ is said to be  $\Xi$-{\em recognizable}
(or  $\+A$-{\em recognizable})
if there exists a finite algebra $\+M=\langle M;\Theta\rangle$ with the same signature  as $\+A$
and a surjective morphism $\varphi:\+A\to \+M$ such that 
such that $B =  \varphi^{-1}(\varphi(B))$, i.e., $B = \varphi^{-1}(T)$ for some subset $T$ of $M$.
\end{definition}

Recognizability can also be stated in terms of congruences.

\begin{lemma}\label{l:rec and congru}
Let $B$ be a subset  of $A$. The following are equivalent

1)  $B$  is $\+A$-recognizable, 

2)  $B$  is saturated with respect to some finite index congruence of $\+A$,

3) the syntactic congruence $\sim^s_B$ of $B$ has finite index.
\end{lemma}
\begin{proof} 1) $\Leftrightarrow$ 2). Let $B$ be saturated with respect to some congruence defined by $\kernel(\varphi)$ (cf. Lemma~\ref{CongVsHomo}), then $B =  \varphi^{-1}(\varphi(B))$. As   $\kernel(\varphi)$ has finite index, $M=\varphi(A)$ is finite and $B$ is recognizable. Conversely, if $B$ is recognizable,  $B =  \varphi^{-1}(\varphi(B))$ is saturated with respect to the congruence  $\kernel(\varphi)$. As $\varphi:\+A\to \+M$ with $M$ finite,  this congruence  has a finite number of classes hence a finite index.

2) $\Rightarrow$ 3). Assume $B$ is saturated for the finite index congruence $\equiv$.
By Proposition~\ref{p:syntactic largest}, $\equiv$ refines $\sim^s_B$ hence $\sim^s_B$ also has finite index.

3) $\Rightarrow$ 2). Follows from the fact that $B$ is saturated with respect to $\sim^s_B$.
\end{proof}
\begin{remark}\label{rk:rr}
Recall the difference between the notions of {\it recognizable} and {\it regular} subsets for a monoid $X$:
a subset $L$ of $X$ is regular if  it can be generated from  finite subsets of $X$ by unions, products and stars. It happens that the two notions coincide, e.g., in $\langle \N;+\rangle$ or in the free monoids.
\end{remark}
%
\subsection{Lattices and Boolean algebras of subsets closed under preimage}
We denote by $\+P(X)$ the class of subsets of $X$.
\begin{definition}
1. A {\em   lattice}  (resp. {\em  complete lattice}) $\+L$ {\em of subsets} of a set $E$ is a family of subsets of $E$
such that $L\cap M$ and $L\cup M$ are in $\+L$ whenever $L,M\in \+L$ (resp. such that any nonempty but possibly infinite) union or intersection of subsets in $\+L$ is in $\+L$.

$\+L$ is a {\em Boolean algebra}  (resp. {\em  complete Boolean algebra})
if it is a lattice  (resp. complete lattice) also closed under complementation.

2. For  $f:E\to E$,
a lattice $\+L$ of subsets of $E$ is {\em closed under} $f^{-1}$ if
$f^{-1}(L)\in\+L$ whenever $L\in\+L$.
\end{definition}

\begin{definition} \label{def:Lat-LAL}For $\+A=\langle A;\Xi\rangle$  an algebra  and $L\subseteq A$, we denote by $\+L_{\+A}(L)$ (resp. $\+L_{\+A}^\infty(L)$\;)  the smallest sublattice (resp. complete sublattice)   $\+L$ of $\+P(A)$ containing $L$ and closed under the inverses of the {\textit{DUO}}s:  i.e.,  $ \gamma^{-1}(Z)\in\+L$, for all $Z\in\+L$, for all $\gamma\in {\textit{DUO}}$.

 We  denote by $\+B_{\+A}(L)$ (resp. $\+B_{\+A}^\infty(L)$\;) the Boolean algebra
(resp. complete Boolean algebra) similarly defined.
\end{definition}
\begin{example} 1) If $\+A=\langle \N;\suc\rangle$,  $\+L_{\+A}(L)$ is the smallest lattice 
containing $L$ and closed under $(x\mapsto x+1)^{-1}$, i.e., closed under decrement where $(L-1)=\{n-1\mid n\in L,n-1\geq 0\}\in\+L$, e.g.,  $\{0,3,7\}-1= \{2,6\}$.

2) If $\+A'=\langle \N;+\rangle$,  $\+L_{\+A'}(L)$ is the smallest lattice 
containing $L$ and closed under $(x\mapsto x+a)^{-1}$ for all $a\in \N$. Since this last closure amounts to closure under decrement,
we have $\+L_{\+A}(L) = \+L_{\+A'}(L)$.

3) If $\+A''=\langle \N;\times\rangle$,  $\+L_{\+A''}(L)$ is the smallest lattice 
containing $L$ and closed under $(x\mapsto ax)^{-1}$, 
i.e.,  the set $L/a= \{n\mid an\in L\}\in\+L$.
 For instance  $\{0,3,7\}/3= \{0,1\}$.

4) If $\+S=\langle\Sigma^*; \cdot \rangle$, $\+L_{\+S}(L)$ is the smallest lattice 
containing $L$ and closed under  $(x\mapsto w\cdot x\cdot w')^{-1}$ for $w,w'\in\Sigma^*$, i.e.,  the set $w^{-1}L{w'}^{-1}= \{x\mid w\cdot x\cdot w'\in L\}\in\+L$.
\end{example}

\begin{lemma}[Disjunctive Normal Form]
\label{l:normal LAL}
1)  Every set in $\+L_{\+A}^\infty(L)$ \big(resp. $\+L_{\+A}(L)$\big)  is of the form
$\cup_{i\in I} \big(\cap_{\gamma\in \Gamma_i}\gamma^{-1}(L)\big)$
where the $\Gamma_i$'s are subsets of $\DUO$ \big(resp. with $I$ and the $\Gamma_i$'s  finite\big).

2) Every set in $\+B_{\+A}^\infty(L)$ \big(resp. $\+B_{\+A}(L)$\big)  is of the form
$\cup_{i\in I} \big(\cap_{\gamma\in \Gamma_i}\gamma^{-1}(L_{i,\gamma})\big)$
where $L_{i,\gamma}$ is either $L$ or its complement $A\setminus L$,
and the $\Gamma_i$'s are as above.
\end{lemma}

\begin{proof}
1) As $\cap$ and $\cup$ distribute over each other, 
 and $\gamma^{-1}(\cup_{i\in I} L_i)= \cup_{i\in I} \gamma^{-1}(L_i)$, and similarly for $\cap$, every arbitrary (resp. finite) $\cap,\cup$ combination of the $\gamma^{-1}(L)$'s, 
$\gamma\in\DUO$, can be put in a disjunctive normal form of the mentioned type. 
 A similar argument proves 2).
\end{proof}

\begin{lemma}\label{bool:saturated}
Let $\+A$ be an algebra and $L$ be a subset of $\+A$.

1) Let $\sim$ be an $\+A$-congruence.
If $L$ is $\sim$-saturated then so is every set in  $\+B_{\+A}^\infty(L)$.

2) Let $\leq$ be an $\+A$-stable preorder. 
If $L$ is a $\leq$-initial segment then so is every set in  $\+L_{\+A}^\infty(L)$.
\end{lemma}
\begin{proof}
{\seg 1) }Since $\sim$ is an $\+A$-congruence, if $x\sim y$ then $\gamma(x)\sim\gamma(y)$
for every $\gamma\in\DUO$.
In particular, if $L$ is $\sim$-saturated 
 then so is $\gamma^{-1}(L)$.
Since $\sim$-saturation is closed under finite or infinite Boolean operations
we conclude that all sets in $\+B_{\+A}^\infty(L)$ are $\sim$-saturated.

 2) Argue similarly, observing that the family of initial segments is closed 
under finite or infinite union and intersection.
\end{proof}

\begin{lemma}\label{l:hull syntactic}
Let $L$ be a subset of an algebra $\+A$.
\\
1. The boolean algebra $\+B_{\+A}^\infty(L)$ is the family of subsets of $A$
which are saturated for the syntactic congruence $\sim_L$ of $L$.
\\
 2. The lattice $\+L_{\+A}^\infty(L)$ is the family of subsets of $A$
which are initial segments for the syntactic preorder $\leq_L$ of $L$.
\end{lemma}
\begin{proof}
 1) As $L$ is $\sim_L$-saturated, Lemma~\ref{bool:saturated} insures that every set 
in $\+B_L^{\infty}(L)$ is also $\sim_L$-saturated.
Conversely, for every element $x\in A$ the $\sim_L$-congruence class of $x$ 
belongs to $\+B_\+A(L)$ since it is equal to
\begin{multline*}
\{y\mid \forall\gamma\in\DUO(\+A)\ (\gamma(x)\in L\Leftrightarrow\gamma(y)\in L)\}
\\=\bigcap_{\gamma\in\DUO(\+A)}
\text{If $\gamma(x)\in L$ then $\gamma^{-1}(L)$ else $A\setminus\gamma^{-1}(L)$}
\end{multline*}
Finally, a $\sim_L$-saturated set is a union of $\sim_L$-congruence classes,
hence also belongs to $\+B_\+A(L)$.

2) As $L$ is a $\leq_L$-initial segment, Lemma~\ref{bool:saturated} insures that every set 
in $\+L_L^{\infty}(L)$ is also a $\leq_L$-initial segment.
Conversely, for every element $x\in A$ the $\leq_L$-initial segment $I_x=\{y\mid y\leq_L x\}$ 
belongs to $\+L_\+A(L)$ since it is equal to
\begin{multline*}
\{y\mid \forall\gamma\in\DUO(\+A)\ (\gamma(x)\in L\Rightarrow\gamma(y)\in L)\}
=\bigcap_{\gamma\in\DUO(\+A),\ \gamma(x)\in L} \gamma^{-1}(L)
\end{multline*}
Finally, a $\leq_L$-initial segment $X$ is the union of the $I_x$'s for $x\in X$
hence also belongs to $\+L_\+A(L)$.
\end{proof}

\begin{lemma}\label{l;recFinite}
 If $L$ is a recognizable subset of $\+A$
then $\+L_{\+A}(L)$ and $\+B_{\+A}(L)$ are finite hence
$\+L_{\+A}^\infty(L)=\+L_{\+A}(L)$ and $\+B_{\+A}^\infty(L)=\+B_{\+A}(L)$.
\end{lemma}

\begin{proof}
If $L$ is recognizable, then $\sim_L$ has a finite  index $k$, there are $k$ congruence classes 
and, as each $\gamma^{-1} (L)$  and each $A\setminus \gamma^{-1} (L)$, 
for $\gamma\in {\textit{DUO}}(\+A)$, 
is a union of congruence classes of  $\sim_L$  (cf. Lemma~\ref{syntacticL_congruence}),  
there are at most $2^{k}$  sets $\gamma^{-1} (L)$ 
and $A\setminus \gamma^{-1} (L)$.
Thus, the Boolean algebra $\+B_{\+A}(L)$ is finite 
hence it is complete and equal to $\+B_{\+A}^\infty(L)$.
A fortiori, the lattice $\+L_{\+A}(L)$ is finite 
hence it is complete and equal to $\+L_{\+A}^\infty(L)$.
\end{proof}

\begin{proposition}\label{p:lattice to BA}
If a lattice $\+L$ of subsets of $E$ is closed under $f^{-1}$ 
then so is the Boolean algebra $\+B$ of subsets of $E$ generated by $\+L$.

As a consequence, in subsequent sections, 
every result of the form ``$\+L_{\+A}(L)$ (resp. $\+L_{\+A}^\infty(L)$\;) is closed under $f^{-1}$''
implies its twin statement   ``$\+B_{\+A}(L)$ (resp. $\+B_{\+A}^\infty(L)$\;) is closed under $f^{-1}$''.
\end{proposition}

\subsection{Generated sets}
A convenient generalization of  condition  $f(a)\geq a$ in (1) of Theorem~\ref{thm:ipl1} to arbitrary algebras, consists in assuming that $f$ is such that, for each $a\in A$, $f(a)$ is in the set $\gen(a)$ generated by $\{a\}$ using all functions in $ {\textit{DUO}}(\+A)$ (cf. Definition \ref{def:gen}).
\begin{definition} \label{def:gena} [Generated set] 
For $\+A=\langle A;\Xi\rangle$ an algebra and $a\in A$, let  $\gen(a)$ be the subset of $A$ defined by $\gen(a)=\{\gamma(a)\mid \gamma\in {\textit{DUO}}(\+A)\}$.
\end{definition}
\begin{example} 1) For $\+A=\langle \N;\suc\rangle$ and $\+A'=\langle \N;+\rangle$, we have $\gen(a)=\{b\mid b\geq a\}=a+\N$. Hence $b\in \gen(a)$ if and only if $b\geq a$.  In particular, $f(a)\in gen(a)$ is equivalent to $f(a)\geq a$.

2)  If $\+A''=\langle \N;\times\rangle$,   we have $\gen(a)=\{b\mid a \hbox{ divides } b\}=a\N$.  Thus $b\in \gen(a)$ if and only if $a$ divides $b$. In particular, $f(a)\in gen(a)$ is equivalent to $a$ divides $f(a)$.

3) For $\+S=\langle\Sigma^*; \cdot \rangle$, the algebra of words with concatenation,  $\gen(a)=\{w\cdot a\cdot w'\mid w,w'\in\Sigma^*\}$. In particular, $f(a)\in gen(a)$ is equivalent to $a$ is a factor of $f(a)$.
\end{example}
\begin{remark}
 The failure of the extension of Theorem~\ref{thm:ipl1} to some simple algebras can be related 
to the failure of the hypothesis $f(a)\in \gen(a)$ for every $a\in A$. 
Consider the algebra $\+A=\langle \{a,b\};Id\rangle$
 and $f$ such that $f(a)=b$ and $f(b)=a$.  On the one hand, the sole congruences on $\+A$  are the two trivial ones and $f$ is trivially congruence preserving, even though $f$ fails the  condition $f(x)\in\gen(x)$ as $f(a)=b\notin\gen(a)=\{a\}$.
 On the other hand, letting $L=\{a\}$, the set $f^{-1}(L)=\{b\}$ is not in the lattice 
 $\+L_{\+A}(L)=\+L_{\+A}^\infty(L)= \{\{a\}\}$.
  \end{remark}
%

\section{Case of  natural integers}
\label{ss:frying pan}
We now reinterpret Theorem \ref{thm:ipl1} using the notions introduced in Section \ref{s:def}. Let us first recall ``folk" results about congruences and recognizable sets of $\langle  \N ;+\rangle$.
\subsection{Congruences on $\langle\N;+\rangle$ and $\langle\N;+,\times\rangle$}\label{s:CN+}
\begin{lemma} \label{folk0} A  congruence $\sim$  on $\langle  \N ;Suc\rangle$ or on $\langle  \N ;+\rangle$ is either  equality, or $\sim_{a,k}$ for some $a,k\in\N$, $k\geq1$ where $\sim_{a,k}$  is defined by
\begin{equation}\label{equation:ak}
x\sim_{a,k } y \mbox{  if and only if  }\begin{cases} \mbox{ either } x=y\\
\mbox{or }\  a\leq x\ ,\  a\leq y \mbox{ and }x\equiv  y\pmod k
\end{cases}.
\end{equation}
The congruence $\sim_{a,k}$ has finite index $a+k$.
It is cancellable if and only if $a=0$.
\end{lemma}
\begin{proof} Let  $\equiv$ be a congruence on $\langle \N;+\rangle$ (or on $\langle  \N ;Suc\rangle$) which is not the identity: there are $a$ and $k>0$ such that $a\equiv a+k$. Choose the least (in lexicographic order) such $a,k$; then for all $j$, $(a+j)\equiv(a+j+k)$, hence $x\sim_{a,k } y$ implies $x\equiv y$. 

Moreover, all elements in $\{0,\ldots,a+k-1\}$ are pairwise nonequivalent modulo $\equiv$. 
First, if $0\leq x<a$ and $x < y $, then $x$ and $y$ cannot be equivalent modulo $\equiv$ as $a$ is the least one such that $a\equiv(a+k)$ for some $k$.  
 Finally, we show that if  $a \leq x<y<a+k$, then we also have $x\not\equiv y$.
Indeed, assume by contradiction that $x\equiv y$ and let $h=x-a$ and $\ell=y-x$. 
We then have $0\leq h< h+\ell<k$, $\ell>0$ and $a+h=x\equiv y=a+h+\ell<a+k$. 
As $\equiv$ is a +-congruence, we have $a+h+j\equiv a+h+\ell+j$ for all $j$.
Letting $j=k-(h+\ell)$ yields $a+h+j\equiv a+h+\ell+j=a+k$, hence, as $a\equiv a+k$,  by transitivity of $\equiv$, we get $a\equiv a+h+j$. As $h+j=k-\ell<k$, this contradicts the minimality of $k$.

If $a=0$ then  $\sim_{0,k}$ is the usual congruence modulo $k$ hence it is cancellable.
If $a\geq1$ then $a-1+k\sim_{a,k}a-1+2k$ but $a-1\not\sim_{a,k}a-1+k$
 hence $\sim_{a,k}$ is not cancellable.
  \end{proof}

A priori, congruences, recognizability 
strongly depend upon the signature. However, due to the properties of addition and multiplication on the integers in $\N$  we have

\begin{corollary}\label{+CPimpliesXCPN} 
The three structures 
$\langle \N;+,\times\rangle$, $\langle \N;+\rangle$ and $\langle \N;\Suc\rangle$
yield the same notions of congruence (namely, equality and the $\sim_{a,k}$'s),
congruence preserving function $\N\to \N$ and recognizable subset of $\N$.
\end{corollary}

\begin{proof}
 Every congruence for $\langle \N;+,\times\rangle$ is a fortiori
a congruence for $\langle \N;+\rangle$.
Conversely, observe that the $\sim_{a,k}$'s are stable under multiplication, 
a straightforward property of modular congruences. 
Using Lemma~\ref{folk0}, this shows that every $+$-congruence is also a $\times$-congruence. 
The assertion about congruence preservation is a trivial consequence.
For that about recognizability, use Lemma~\ref{l:rec and congru}.
\end{proof}

 For morphisms  the situation is more complex, as shown by the next Remark.
%
\begin{remark}\label{r:morphimes+x}
 Homotheties $x\mapsto kx$ for $k\geq 2$ are +-morphisms which are not $\times$-morphism
as $k(x\times y)\neq kx\times ky$.
\end{remark}
%

\subsection{The ``frying pan''  monoids and semirings}\label{s:FryingPan}
We define canonical representations of the quotient monoids and semirings
$\N/\!\!\sim_{a,k}$.

 \begin{definition}\label{def:finite monogenic monoids}
Let $a,k\in\N$ such that $k\geq1$.

1) We denote by $M_{a,k}=\{0,\ldots,a+k-1\}$ the set of minimum representatives 
of the equivalence classes of $\sim_{a,k}$
and by $\varphi_{a,k}:\N \to M_{a,k}$ the map such that 
$$
\varphi_{a,k}(x)= \textit{IF } x<a \textit{ THEN } x \textit{ ELSE } a+((x-a)\pmod k)
$$
which can be identified to the canonical surjection $\N\to\N/\!\!\sim_{a,k}$. 

2) To any $n$-ary operation $\xi:\N^n\to\N$ on $\N$ corresponds a unique operation
$\xi_{a,k} : M_{a,k}^n\to M_{a,k}$  making $\varphi_{a,k}$ a morphism 
$\langle\N;\xi\rangle \to \langle M_{a,k};\xi_{a,k}\rangle$;  it is
defined by $\xi_{a,k} (x_1,\ldots,x_n)= \varphi_{a,k}\big(\xi (x_1,\ldots,x_n)\big)$.
In this way, we shall consider the arithmetic operations 
$\Suc_{a,k}$, $+_{a,k}$ and $\times_{a,k}$ on $M_{a,k}$. 
\end{definition}

\begin{definition}
A monoid $\langle M;\oplus\rangle$ with unit $0$ is monogenic if there exists $g\in M$
such that every element of $M\setminus\{0\}$ is a sum of some nonempty finite set
of copies of $g$.
Such an element $g$ is called a {\em generator}.
\end{definition}

\begin{figure}[h]
\[
\scalebox{.70}{ 
\xymatrix{
&&&&&&{a+2}\ar@/^/[dr]&&\\
&&&&&{a+1}\ar@/^/[ru]&&{a+3}\ar@/^/[rd]&\\
\qquad        0\ar  [r] &   1\ar  [r]& 2\ \ldots\!\!\!\!\!\!\!\!\!\!\!\!
&   {a-1}\ar  [r]&   {a}\ar@/^/[ru]&&&&{a+4}\ar@/^/[dl]\\
&&&&&{a+7}\ar @/^/ [lu]&&{a+5}\ar @/^/[ld]&\\
&&&&&&{a+6}\ar @/^/[lu]&&\\
}
}
\]
\caption{``Frying pan'' monoid $M_{a,k}$,  $k=8$, where  $Suc_{a,k}$  is represented by  arrow.
}\label{fig:frying pan}
\end{figure}
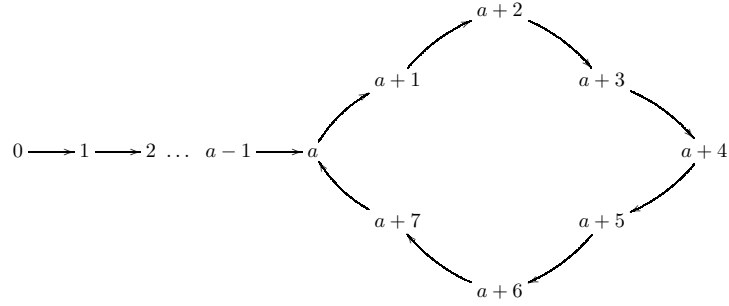
\begin{lemma}\label{l:finite monogenic monoids}
1) $\langle M_{a,k};+_{a,k}\rangle$ is a monogenic commutative monoid 
(called ``frying pan'' monoid, cf. Figure~\ref{fig:frying pan}) with $0$ as unit.

2) Every finite monogenic monoid $\langle M;\oplus\rangle$ is isomorphic to
the monoid $\langle M_{a,k};+_{a,k}\rangle$ for some $a,k$.

3) For every surjective morphism $\psi : \langle \N;+\rangle \to \langle M;\oplus\rangle$
onto a finite monoid $\langle M;\oplus\rangle$, there exists $a,k$
and an isomorphism $\theta : \langle M_{a,k};+_{a,k}\rangle \to \langle M;\oplus\rangle$
such that $\psi = \theta \circ \varphi_{a,k}$.
\end{lemma}
\begin{proof}
We recall the argument of the classical proof of 2) which is also used
for Lemma~\ref{l:generators monogenic}.
Let $g$ be a generator of $M$.
Consider the relation on $\N$ such that $\ell\equiv n$ if
the sums in $\langle M;\oplus\rangle$ 
of $\ell$ copies of $g$ and that of $n$ copies of $g$ are equal.
This relation $\equiv$ is a congruence on $\langle\N;+\rangle$ and it has finite index
since $M$ is finite.
Thus, it is equal to $\sim_{a,k}$ for some $a,k$.
The wanted isomorphism $\theta : \langle M_{a,k};+_{a,k}\rangle \to \langle M;\oplus\rangle$
maps $x\in M_{a,k}$ onto the sum in $M$ of $x$ copies of $g$.
To get 3) observe   that $\langle M;\oplus\rangle$ is necessarily monogenic as so is $\langle\N;+\rangle$, and that the image $g=\psi(1)$ of the generator $1$ of $\langle\N;+\rangle$
is a generator of $\langle M;\oplus\rangle$,  then use the above isomorphism $\theta$.
\end{proof}

\begin{lemma}\label{l:generators monogenic}
The generators of the monogenic monoid $\langle M_{a,k};+_{a,k}\rangle$ are as follows:
\begin{itemize}
\item
If $a\geq 2$  then $1$ is the unique generator,
\item
If $a\in\{0,1\}$ the generators are the elements of $\{1,\ldots,a+k-1\}$ which are coprime  with $k$.
\end{itemize}
\end{lemma}
\begin{proof}
Obviously, $1$ is a generator in all cases.
Let $g$ be a generator. Necessarily $g\neq0$.
If $a\geq2$ since any $+_{a,k}$ sum of copies of an element $\geq2$ is also $\geq2$,
the sole way to obtain $1$ as a $+_{a,k}$ sum of copies of $g$ is that $g=1$.
If $a=0$ then the set of $+_{0,k}$ sums of nonempty finite sets of copies of $g$ is equal to 
$\{ng\pmod k \mid n\geq1\} = \{ng\pmod k \mid n=1,\ldots,k\}$ 
and contains $M_{0,k}\setminus\{0\}$
if and only if $g$ is  coprime with $k$. 
If $a=1$ then, for $g\geq1$,
the set of $+_{1,k}$ sums of nonempty finite sets of copies of $g$ is equal to 
$\{1+(ng-1\pmod k) \mid n\geq1)\} = \{1+(ng-1\pmod k) \mid n=1,\ldots,k\}$ 
and contains $M_{a,k}\setminus\{0\}$
if and only if $g$ is  coprime with  $k$. 
\end{proof}

\begin{lemma} \label{+morphisms} 
There is a bijective correspondence between the generators of a finite monogenic monoid
$\langle M;\oplus\rangle$ and
the surjective morphisms $\langle\N;+\rangle\to \langle M;\oplus\rangle$, defined by
\begin{eqnarray}\label{eq:psi g}
g &\leadsto& \textit{ the unique morphism 
$\psi_g : \langle\N;+\rangle\to \langle M;\oplus\rangle$  such that } \psi_g(1)=g
\end{eqnarray}
\end{lemma}
\begin{proof} 
Given a generator $g$ of $M$,  let $\psi_g$ be defined by condition \eqref{eq:psi g} 
together with $\psi_g(0)=0_M$  and, for $n\geq1$, 
$\psi_g(n)=\overbrace{g\oplus\cdots\oplus g}^{\text{$n$ times}}$ ;
$\psi_g$ defines a surjective morphism.
Conversely, as $1$ is a generator of $\langle\N;+\rangle$,  if $\psi : \langle\N;+\rangle\to \langle M;\oplus\rangle$ is a surjective morphism,
then  $\psi(1)$ is a generator of $M$.
\end{proof}
 \begin{corollary}\label{cor:+{ak}morphisms}
If $a\geq2$ or $(a,k)\in\{(0,1),(0,2),(1,1),(1,2)\}$ 
then $\varphi_{a,k}$ is the unique surjective morphism
$\langle\N;+\rangle\to \langle M_{a,k},+_{a,k}\rangle$.
If $a\in\{0,1\}$ and $k\geq3$,  then there are $\phi(k)$ distinct  
surjective morphisms $ \langle\N;+\rangle\to \langle M_{a,k};+_{a,k}\rangle$,
where $\phi\colon\N\to\N$ is Euler totient function mapping $x$ to the number of integers $\leq x$
which are coprime with $x$.
\end{corollary}
\begin{proof}It follows from Lemma~\ref{l:generators monogenic}.
\end{proof}
\begin{definition}\label{def:semiring}
A {\em semiring} is a set $R$ equipped with two binary operations $\oplus$ and $\otimes$ such that
\\- $\langle R, \oplus\rangle$ is a commutative monoid with an identity element, say $0$,
\\- $\langle R, \otimes\rangle$ is a monoid with an identity element,
\\- Multiplication by 0 annihilates $R$: $0\otimes a = a\otimes 0 = 0$ for all $a$,
\\- Multiplication left and right distributes over addition: 
$a\otimes(b\oplus c)=(a\otimes b)\oplus(a\otimes c)$
and $(a\oplus b)\otimes c=(a\otimes c)\oplus(b\otimes c)$ for all $a,b,c$.
\end{definition}

We use the arithmetic operations defined on $M_{a,k}$, 
cf. Definition \ref{def:finite monogenic monoids}.

\begin{lemma}\label{l:varphi ak morphism semiring}
The algebra $\langle M_{a,k};+_{a,k},\times_{a,k}\rangle$ is a semiring,
called the ``({a,k}) frying pan semiring'',
and $\varphi_{a,k}$ is a morphism 
$\langle\N;+,\times\rangle \to \langle M_{a,k};+_{a,k},\times_{a,k}\rangle$. 
\end{lemma}
\if 35

{\ire CE LEMME NE SERT NULLE PART{\color{red} SUPPRIMER ??}\\
\sege NON, CAR IL SERT A REPONDRE A LA QUESTION : LES +-MORPHISMES
SONT-ILS AUSSI DES $(+,\times)$-MORPHISMES ?}
{\seg

The following Lemma is the key to complete Corollary~\ref{+CPimpliesXCPN}
by fully explaining the reason behind the counterexample in 
Remark~\ref{r:morphimes+x}.

\begin{lemma} \label{l:monogenic semiring} 
Let $\langle M;\oplus,\otimes\rangle$ be a finite monogenic semiring.
The unit element of the multiplicative monoid $\langle M;\otimes\rangle$ is a generator of 
the additive monoid $\langle M;\oplus\rangle$.
\end{lemma}
\begin{proof} 
The case where $M$ is a singleton set is trivial.
We now assume that $M$ has at least two elements.
Using Lemma~\ref{l:finite monogenic monoids}, we reduce to the case where the additive 
monoid $\langle M;\oplus\rangle=\langle M_{a,k};+_{a,k}\rangle$ with $k\geq1$ and $a+k\geq2$.
Let $u$ be the unit of $\otimes$.
Observe that $u\neq0$ : else, we would have $x=x\otimes u=x\otimes 0=0$, 
contradicting the hypothesis that $M$ has $n\geq2$ elements.
Since $u\geq1$ and $1$ is a generator of $\langle M;\oplus\rangle$,  
$u$ is the $+_{a,k}$ sum of $u$ copies of $1$;   as $u$ is the unit of $\otimes$, we have 
 $1=u\otimes 1=(1+_{a,k}\cdots+_{a,k}1)\otimes 1=
(1\otimes1)+_{a,k}\cdots+_{a,k}(1\otimes1)$ 
so that $1$ is the sum of $u$ copies of $1\otimes1$.
As $1$ is a generator and $1\otimes1$ generates $1$ we see that $1\otimes1$ is a generator.

If $1\otimes1=1$ then $1$ is the $+_{a,k}$ sum of $u$ copies of $1$. 
This implies that $u=1$ hence is a generator.

If $1\otimes1\neq1$ then $\langle M_{a,k};+_{a,k}\rangle$ has generators other than $1$
and Lemma~\ref{l:generators monogenic} insures that $a\in\{0,1\}$ 
and $1\otimes1$ is coprime to $k$.
If $a=0$, the fact that $1$ is the $+_{0,k}$ sum of $u$ copies of $1\otimes1$ means that
$1=u(1\otimes1)\pmod k$ and implies that $u$ is also coprime with $k$ hence $u$ is a generator
of $\langle M_{0,k};+_{0,k}\rangle$.
If $a=1$, the fact that $1$ is the $+_{1,k}$ sum of $u$ copies of $1\otimes1$ means that
$1=1+((u(1\otimes1)-1)\pmod k)$, i.e. $u(1\otimes1)=1\pmod k$.
Again, this implies that $u$ is also coprime with $k$ hence $u$ is a generator
of $\langle M_{1,k};+_{1,k}\rangle$.
\end{proof} 

In contrast with Corollary~\ref{+CPimpliesXCPN},
there is a unique surjective morphism $\langle\N;+\rangle\to \langle M_{a,k};+_{a,k}\rangle$
which is also a $\langle\N;\times\rangle\to \langle M_{a,k},\times_{a,k}\rangle$ morphism.

\begin{lemma} \label{+xmorphisms} 
Let $\langle M;\oplus,\otimes\rangle$ be a finite monogenic semiring.
There exists a unique morphism $\psi : \langle\N;+\rangle\to\langle M;\oplus\rangle$
which is also a morphism $\psi : \langle\N;+,\times\rangle\to\langle M;\oplus,\otimes\rangle$.
namely the one which maps $1$ onto the unit element of $\langle M;\otimes\rangle$.
In particular, $\varphi_{a,k}$ is the unique morphism 
$\langle\N;+,\times\rangle\to \langle M_{a,k},+_{a,k}\times_{a,k}\rangle$.
\end{lemma}
\begin{proof} 
Since $1$ generates $\langle\N;+\rangle$, any additive morphism between $\N$ and $M$
is entirely determined by the image of $1$.
Since $1$ is the unit of $\langle\N;\times\rangle$,
its image by any multiplicative surjective morphism
$\langle\N;\times\rangle\to\langle M;\otimes\rangle$ is the unit of $\langle M;\otimes\rangle$.
This shows that there is at most one surjective morphism
$\langle\N;+,\times\rangle\to\langle M;\oplus,\otimes\rangle$.
The existence of such a morphism is a consequence of Lemmas~\ref{l:finite monogenic monoids} and \ref{l:varphi ak morphism semiring}.
\end{proof}
}
\fi
\subsection{Congruence preservation and divisibility}\label{ss:CPD}
Congruence preservation on $\langle\N;+\rangle$ can be characterized as follows
\begin{theorem}\label{thm:CPsurN}
For a map $f:\N\to\N$, the following conditions are equivalent
\begin{enumerate}
\item
$f:\N\to\N$ is congruence preserving on the algebra  $\langle\N;+\rangle$,
\item
$\left\{
   \text{\begin{tabular}{l}
(i) \qquad\ $(x-y)$ divides $(f(x)-f(y))$ for all $x,y\in\N$, and \label{CPa}\\
(ii) \quad\ \  either $f$ is constant or $f(x)\geq x$ for all $x$. \label{CPb}
\end{tabular}}
\right.$
\end{enumerate}
\end{theorem}

\begin{proof}
$(1)\Rightarrow(2)$.
Suppose $f$ is congruence preserving. 
Let $x<y$ and consider the congruence modulo $y-x$.
As $y\equiv x\bmod y-x$ we have $f(y)\equiv f(x)\bmod y-x$ hence $y-x$ divides $f(y)-f(x)$.
This proves (i).
To prove (ii), we show that if condition $f(x)\geq x$ is not satisfied then $f$ is constant.
Let $a$ be least such that $f(a)<a$.
 Consider 
the frying pan $M_{a,1}$ (depicted in Figure~\ref{Ma,1}) 
and the congruence $\sim_{a,1}$. 
\begin{figure}[h]
\[\xymatrix{
0\ar[r]&1\ar[r]&\ldots &{a-1}\ar[r]&a\ar@(ul,ur)
}\]
\caption{$M_{a,1}$}\label{Ma,1}
\end{figure}
We have $a\sim_{a,1} y$ for all $y\geq a$ hence $f(a)\sim_{a,1}f(y)$. 
As $f(a)<a$ this implies $f(a)=f(y)$.
Thus, $f(a)=f(y)$ for all $y\geq a$.
Let $z<a$. By condition (i) (already proved) we know that $p$ divides $f(z)-f(z+p)$ for all $p$.
Now, $f(z+p)=f(a)$ when $z+p\geq a$. Thus, $f(z)-f(a)$ is divisible by all $p\geq a-z$.
This shows that $f(z)=f(a)$.
Summing up, we have proved that $f$ is constant with value $f(a)$.

$(2)\Rightarrow(1)$.
Constant functions are trivially congruence preserving. 
We thus assume that $f$ is not constant, hence $f$ satisfies condition (i) and $f(x)\geq x$ for all $x$.
Consider a congruence $\sim$ and suppose $x\sim y$. 
If $\sim$ is the identity relation then $x=y$ hence $f(x)=f(y)$.
Else, by Lemma~\ref{folk0} the congruence $\sim$ is $\sim_{a,k}$ with $a\in\N$ and $k\geq1$.
In case $y<a$ then condition $x\sim_{a,k} y$ implies $x=y$ hence $f(x)=f(y)$ and $f(x)\sim f(y)$.
In case $y\geq a$ then condition $x\sim_{a,k} y$ implies $x\geq a$ and $x\equiv y\bmod k$,
hence $k$ divides $y-x$.
Condition (i) insures that $y-x$ divides $f(y)-f(x)$ hence $k$ also divides $f(y)-f(x)$.
Also, our hypothesis yields $f(x)\geq x$ and $f(y)\geq y$.
As $x,y\geq a$ we get $f(x),f(y)\geq a$. Since $k$ divides $f(y)-f(x)$ we conclude that 
$f(x)\sim_{a,k} f(y)$, whence (1).
\end{proof}
\begin{remark}\label{exCPnonsurlineaire} 
 We cannot withdraw the over-linearity condition $f(x)\geq x$ in Theorem \ref{thm:CPsurN}. We proved in {\rm\cite{cgg15}} that any function $f\colon\Z/n\Z\to
\Z/n\Z$ satisfying (2)$(i)$ can be lifted to a function $F\colon\N\to\N$  such that $F$ satisfies (2)$(i)$ and for $x\leq n-1$, $F(x)=f(x)$. Consider the function $f\colon\Z/3\Z\to\Z/3\Z$  such that $f(0) = f(1)=0<1$, $f(2)=2$;  $F$ is non constant, satisfies (2)$(i)$ but is not congruence preserving as $F(1)<1$  (using Theorem \ref{thm:CPsurN}). In fact we can  directly see that 
$F$ is not congruence preserving using the congruence $\sim_{1,1}$: indeed $1\sim_{1,1}2$
but $F(1)\not\sim_{1,1}F(2)$.
\end{remark}

Using Lemma \ref{lem:de n a 1}, we    immediately  deduce  from Theorem \ref{thm:CPsurN}
\begin{corollary}
Function $f\colon \N^n\longrightarrow \N$ is congruence preserving if and only if condition (2) of Theorem \ref{thm:CPsurN}
holds for all the unary frozen functions $f^{\vec{a}}_i$,  $1\leq i\leq n$ and $\vec{a}\in \N^{n-1}$.
\end{corollary}
%
\subsection{Congruence/(pre)orders preservation and monotonicity}\label{s:34}

Theorem~\ref{thm:CPsurN} may induce the hope that  congruence preserving functions   over $\langle\N;+\rangle$ are  monotone, but this is not the case.
Counterexamples can be obtained using  the following Propostition \ref{p:g idr f} (Theorem 3.15 in  \cite{cgg14}).

\begin{proposition}[\cite{cgg14}]\label{p:g idr f}
For every $f:\N\to\N$ there exists a function $g:\N\to\Z$ such that,
letting $\lcm(x)$ be the least common multiple of $1,\ldots,x$,
\begin{enumerate}
\item[i)]
$x-y$ divides $g(x)-g(y)$ for all $x,y\in\N$,
\item[ii)]
$f(x)- 2^x \lcm(x) \leq g(x)\leq f(x)$.
\end{enumerate}
\end{proposition}
\begin{proposition}\label{p:cp non monotone}
 There exists a non  monotone function $g:\N\to\N$
which is $\langle\N;+\rangle$-congruence preserving.
\end{proposition}
\begin{proof}
Let $f(x) = \sum_{0\leq y\leq x, y \textit{ even}} 2^{y+2}\lcm(y+2)
                 -  \sum_{1\leq z\leq x, z \textit{ odd}} 2^{z}\lcm(z)$
and let $g$ be as in Proposition~\ref{p:g idr f}.
We first prove that $g$ maps $\N$ into $\N$. A simple computation shows that 
$f(x) > 3\,(2^x)\,\lcm(x)$ for all $x\in\N$, 
hence $g(x) \geq f(x) - 2^x\lcm(x) > 2^{x+1}\,\lcm(x)$.
This proves that $g(x)\in\N$ and moreover $g(x)>x$ for all $x$.

We next prove that $g$ is non monotone.
Condition $ii)$ of Proposition \ref{p:g idr f} implies 
\begin{equation}\label{eqn:g idr f}
(f(x)-2^{x} \lcm(x)) - f(x-1)\leq g(x)-g(x-1) \leq f(x) - (f(x-1) - 2^{x-1} \lcm(x-1))
\end{equation}
Computing the lower and upper bounds in \eqref{eqn:g idr f} shows that:
 for $x$ even,  $g(x)-g(x-1) \geq  (f(x)-f(x-1)) -2^{x} \lcm(x) =
2^{x+2}\lcm(x+2) - 2^{x} \lcm(x) >0$, while for $x$ odd,
$g(x)-g(x-1) \leq f(x)-f(x-1) + 2^{x-1} \lcm(x-1)) = -2^x\lcm(x) + 2^{x-1}\lcm(x-1) <0$.
Thus, $g$ is a zig zag function hence non monotone.

Finally, as $g(x)>x$ for all $x$ and $x-y$ divides $g(x)-g(y)$ for all $x,y\in\N$ (by condition $i)$ of Proposition \ref{p:g idr f},  both conditions $(i)$ and $(ii)$  of Theorem \ref{thm:CPsurN} {\it (2)} hold and $g$ is congruence preserving.
\end{proof}


The existence of a stable total order has a nice consequence.

\begin{proposition}\label{p:stable total order}
Let $\+A = \langle A;\Xi\rangle$ be an algebra.
Assume there is a total order on $A$ which is $\+A$-stable.
Then a function $f:A\to A$ is $\+A$-stable preorder preserving 
if and only if it is $\+A$-stable order preserving.
\end{proposition}

\begin{proof}
One implication is trivial. We show that if $f$ preserves all $\+A$-stable orders
then it also preserves all $\+A$-stable preorders.
Let $\leq$ be a total $\+A$-stable order on $A$ and $\preceq$ be an $\+A$-stable preorder on $A$.
Then $\leq\cap\preceq$ and $\geq\cap\preceq$ are $\+A$-stable orders on $A$.
Indeed, as reflexivity and transitivity go through intersection,
the intersection of two preorders is a preorder. Also, the antisymmetry property of an order
(namely, $x\leq y\leq x$ implies $x=y$) still holds for the intersection with any relation.
Finally, the intersection of two $\+A$-stable relations is $\+A$-stable.

Suppose now that $x\preceq y$. As $\leq$ is total, either $x\leq y$ or $y\leq x$.
Assume $x\leq y$. As $f$ preserves the $\+A$-stable order $\leq\cap\preceq$
and $x(\leq\cap\preceq) y$ we have $f(x)(\leq\cap\preceq) f(y)$ and a fortiori
$f(x)\preceq f(y)$.
Same argument if $y\leq x$ using the order $\geq\cap\preceq$.
\end{proof}

\begin{remark}\label{rk:stable total order}
In particular, considering the usual total order on $\RR$,
the above result applies for the algebras $\langle\RR;+\rangle$ and
$\langle[0,+\infty[;+,\times\rangle$ and their subalgebras,
e.g., $\langle\N;+\rangle$ and $\langle\N;+,\times\rangle$.
It also applies to products of these algebras
(consider the stable lexicographic product of the usual order).
\end{remark}

We now show what monotonicity adds to congruence preservation in $\langle\N;+\rangle$ and $\langle\N;+,\times\rangle$.
First, a simple observation.

\begin{dlemma}\label{l:cp and spp on N}
Let $\preceq$ be a stable order on $\langle\N;+\rangle$ and, 
for $a\in\N$, let $M^+_a=\{x \mid a\preceq a+x\}$ and $M^-_a=\{x \mid a+x\preceq a\}$.

1) $M^+_a$ and $M^-_a$ are submonoids of $\langle\N;+\rangle$.

2) If $a\leq b$ then $M^+_a \subseteq M^+_b$ and $M^-_a \subseteq M^-_b$.
\end{dlemma}

\begin{proof}
1) Clearly, $0\in M^+_a$. 
Suppose $a\preceq a+x$ and $a\preceq a+y$.
By stability the first inequality yields $a+y\preceq a+x+y$ and, by transitivity,
the second inequality gives $a\preceq a+x+y$. Thus, $M^+_a$ is a submonoid.
Idem with $M^-_a$.

2) If $a\preceq a+x$ then by stability $a+(b-a)\preceq a+x+(b-a)$, i.e. $b\preceq b+x$.
Thus, $M^+_a \subseteq M^+_b$. Similarly, $M^-_a \subseteq M^-_b$. 
\end{proof}

\begin{theorem}\label{prop:miracle}
Relative to the algebras $\langle\N;+\rangle$ and $\langle\N;+,\times\rangle$,
a function $f:\N\to\N$ is stable preorder preserving if and only if it is monotone non decreasing
and congruence preserving.
\end{theorem}

\begin{proof}
If $f$ is stable preorder preserving then it is a fortiori congruence preserving.
As the usual order on $\N$ is stable, it is preserved by $f$ hence $f$ is monotone non decreasing.

Assume $f$ is congruence preserving and monotone nondecreasing.
We prove that $f$ preserves stable preorders.
The case $f$ is constant is trivial. We now suppose $f$ is not constant hence
$f$ satisfies conditions (i) and (ii) in  Theorem \ref{thm:CPsurN}.
Since the usual order on $\N$ is stable, using Lemma~\ref{l:cp and spp on N} it suffices
to show that $f$ preserves every stable order $\preceq$. 
Let $\sim$ be the congruence associated with $\preceq$.
 Suppose $a\preceq b$, then
\begin{itemize}
\item either  $a\leq b$ hence  $b-a\in M^+_a=\{x \mid a\preceq a+x\}$.
By condition (i) we know that $b-a$ divides $f(b)-f(a)$. 
As $M^+_a$ is a monoid, $f(b)-f(a)$ is also in $M^+_a$.
Condition (ii) insures that $f(a)\geq a$ hence,  by Lemma~\ref{l:cp and spp on N},
$M^+_a \subseteq M^+_{f(a)}$. Thus, $f(b)-f(a)$ is in $M^+_{f(a)}$
implying $f(a)\preceq f(b)$. 
\item or $b\leq a$, the proof is similar, by noting that $a-b\in
M^-_b=\{x \mid b+x\preceq b\}$.\qedhere
\end{itemize}
\end{proof}
%
\subsection{Recognizable subsets of $\langle\N;+\rangle$ and  $\langle\N;+,\times\rangle$}\label{s:35}
 We recall the classical characterization of recognizability in $\langle\N;+\rangle$ and $\langle\N;+,\times\rangle$.
\begin{proposition}\label{p:folk N rec}
Let $L$ be a subset of $\N$. The following conditions are equivalent.
\begin{enumerate}
\item
$L$ is $\langle\N;+\rangle$-recognizable
\item
$L$ is $\langle\N;+,\times\rangle$-recognizable
\item
$L$ is of the form $L=F\cup(R+k\N)$ with
$1\leq k$, $F\subseteq\{x\mid 0\leq x<a\}$,  and $R\subseteq\{x\mid a\leq x<a+k\}$
(possibly empty in which case $L$ is finite).
\end{enumerate}
\end{proposition}
\begin{proof} 
$(1)\Leftrightarrow(2)$. By Corollary~\ref{+CPimpliesXCPN}.

$(1)\Leftrightarrow(3)$.
By definition, a subset $L$ of $\N$ is $\langle\N;+\rangle$-recognizable if it is of the form
$L=\psi^{-1}(U)$ for some surjective morphism 
$\psi : \langle\N;+\rangle \to \langle M;\oplus\rangle$
onto a finite monoid 
and some $U\subset M$. 
 Using  Lemma~\ref{l:finite monogenic monoids} 3),
 we reduce to the case $M=M_{a,k}$ and
$\psi=\varphi_{a,k}$ for some $a,k$.
Letting
$F=U\cap\{0,\ldots,a-1\}$ and $Y=\{x\in\{0,\ldots,k-1\}\mid a+x\in U\}$,
we have $U=F\cup(a+Y)$ and
$\varphi_{a,k}^{-1}(F)=F$ and $\varphi_{a,k}^{-1}(a+Y)=a+Y+k\N=R+k\N$
hence $\varphi_{a,k}^{-1}(U)=F\cup(R+k\N)$.
\end{proof}

In \cite{cgg14ipl} we proved a connection between recognizable subsets and 
functions satisfying conditions {\it (2) (i) } and {\it (2) (ii) } of Theorem \ref{thm:CPsurN}. Using the equivalence given by Theorem \ref{thm:CPsurN} we now can reformulate the result  of  \cite{cgg14ipl} as  the following version of Theorem \ref{thm:ipl1}.
\begin{theorem}\label{thm:ipl} 
Let $f:\N\longrightarrow\N$ be a non decreasing function, 
the following conditions are equivalent:
\begin{enumerate}
\item[$(1)_\N$] 
The function $f$ is $\langle\N;+\rangle$-congruence preserving on $\+\N$ and 
$f(x)\geq x$ for all $x$.
\item[$(2)_\N$] 
For every finite subset $L$  of $\N$, the lattice $\+L_{\langle\N;+\rangle}(L)$ is closed under $f^{-1}$.
\item[$(3)_\N$] 
For every recognizable subset $L$ of $\langle \N;+\rangle$, the lattice $\+L_{\langle\N;+\rangle}(L)$ is  closed under $f^{-1}$.
\end{enumerate}
\end{theorem}
%

\subsection{Case of  $\langle  \N ;\times\rangle$} \label{ss:Ntimes}
We here extend Proposition~\ref{thm:ipl}  to congruence preserving function on $\langle  \N ;\times\rangle$  and  explicitly characterize these functions in Theorem \ref{thm:iplX} $(ii)$. We
have seen that congruences and morphisms of $\langle\N;+\rangle$ coincide with congruences and morphisms of $\langle\N;+,\times\rangle$. The situation changes radically when considering
the algebra $\langle  \N ;\times\rangle$.
\begin{theorem}\label{thm:iplX}
For $f:\N\setminus\{0\}\longrightarrow\N\setminus\{0\}$ 
the following conditions are equivalent:
\begin{enumerate}
\item[(i)] 
For every recognizable subset $L$ of $\langle \N\setminus\!\{0\};\times\rangle$ the lattice $\+L_\times(L)$ is  closed under $f^{-1}$.
\item[(ii)] The function  $f$ is  of the form 
$f(x)=f(1)\times x^n$ for some fixed  $n\in\N$.
\item[(iii)]  The function
$f$ is $\langle\N\setminus\{0\}; \times\rangle$-congruence preserving and  $x$  divides $f(x)$ for all $x$. 
\end{enumerate}
\end{theorem}
It  appears that a sharp difference between Theorem \ref{thm:iplX} and Proposition\ref{thm:ipl} 
is the richness of the family of involved  functions.  For instance on  $\langle \N;+\rangle$ this family contains non polynomial functions. This can be explained by the fact that there are many more congruences on $\langle\N\setminus\{0\}; \times\rangle$ than $\langle \N;+\rangle$. See Example \ref{r:morphimes+xbis}.
\begin{definition} Let $P$ be the set of prime numbers.
For $p\in P$,  the  $p$-valuation of $x$ denoted $\val(x,p)$ is the
highest exponent $n$ of $p$ such that $p^n$ divides $x$.
\end{definition}
 \begin{example}\label{r:morphimes+xbis}
 1) There are strictly more congruences on $\langle \N\setminus\{0\};\times\rangle$ 
than on  $\langle \N;+\rangle$. Let $Q\neq P$ be a nonempty set of prime numbers. The 
 relation $\sim_Q$ such that $x\sim_Q y$ if for all $p\in Q$, $\val(x,p)=\val(y,p)$
is a congruence for $\langle \N\setminus\{0\};\times\rangle$ but not for $\langle \N;+\rangle$.
For instance for $Q=\{2,5\}$, $2\sim_Q 6$, but $4=2+2\not\sim_Q  6+2=8$.
In particular, there are uncountably many $\langle \N\setminus\{0\};\times\rangle$-congruences
whereas there are only countably many $\langle \N;+\rangle$-congruences.

2) The same phenomenon occurs for maorphisms. Let $\varphi_Q\colon\N\to \N$ be such that $\varphi_Q(x)=\prod_{p\in Q} p^{\val(x,p}$. This map
is a $\langle \N\setminus\{0\};\times\rangle$-morphism. There are thus uncountably many $\langle \N\setminus\{0\};\times\rangle$-morphisms
whereas there are only countably many $\langle \N;+\rangle$-morphisms.

2) Another example: the relation $\sim$ such that $x\sim y$ if $x=y$ or both $x,y$ are powers of $2$
is a congruence for $\langle \N;\times\rangle$ but not for $\langle \N;+\rangle$
since $2\sim 4$ and $4\sim 4$ but $2+4=6\not\sim 4+4=8$.
This stems from the fact that a mapping $\varphi\colon \langle \N;+,\times\rangle \to \langle M;\oplus,\otimes\rangle$
can be a $\times$-morphism without being a +morphism and vice versa.

2) Let $\varphi\colon\langle  \N ;+,\times\rangle\to \langle \Z/2\Z;+,\times\rangle$  be defined by: 
$\varphi(x) =1$ if and only if $x=2^n$ for some $n\in\N$ and $\varphi(x) =0$  otherwise; $\varphi$
is a  $\times$-morphism but not a  $+$-morphism because $\varphi(6+1)=0\not= 0+1=\varphi(6)+\varphi(1)$.
\end{example}
Before proving Theorem \ref{thm:iplX}, let us introduce some simple notation and obervationss.
 \begin{notation}
 If $L$ is a subset of $\N$, and $a\in\N$, $L/a$ denotes the set of exact quotients of elements of $L$ by $a$, $L/a=\{x\mid x\in\N {\text{ and \ }}ax\in L\}$.
 \end{notation} 

\begin{lemma}\label{NX} 1) ${\textit{DUO}}(\langle \N ; \times\rangle)=\{ x\mapsto ax\mid a\in\N^*\}$ is the set of homotheties.

2)   For $L\subseteq \N\setminus\{0\}$,  $\+L_\times(L)=\+L_{\langle \N ; \times\rangle}(L)$ is  the smallest sublattice $\+L$ of  $\+P(\N)$ containing $L$ and closed  under exact division of sets where $L/a=\{x\mid x\in\N {\text{ and \ }}ax\in L\}$. In particular, for a singleton set $L=\{c\}$, all sets in $\+L_\times(L)$ are sets of divisors of $c$.
  \end{lemma}

 \begin{example}  1) For $L=\{2^n5^p\mid n,p\in\N\}$, we have $\+L_\times(L)=\{\emptyset,L\}$.

2) For $L=\{2^n5^p\mid n\leq 2, p\leq 1\}=\{1,2,4,5,10,20\}$, the set $\+L_\times(L)$ is

$\big\{\emptyset,\{1\},\{1,2\},\{1,5\},\{1,2,4\},\{1,2,5\},
\{1,2,4,5\}, \{1,2,5,10\}, \allowbreak  \{1,2,4,5,10,20\}\big\}$

\noindent since the sets $L/n$ are  given by the following table:
$$
\begin{array}{c|ccccccc|}
\cline{2-8}
a&1&2&4&5&10&20&\notin L\\\cline{2-8}
L/a&\{1,2,4,5,10,20\}&\{1,2,5,10\}&\{1,5\}&\{1,2,4\}&\{1,2\}&\{1\}&\emptyset
\\\cline{2-8}
\end{array}
$$
\end{example}

\begin{proof} [Proof of Theorem \ref{thm:iplX}]
$(i)\Longrightarrow (ii)$ If $f$ is constant, then  $(ii)$ holds with $n=0$. Let us assume from now on that  $f$ is non constant.

{\bf Fact 1.} We first prove that $(i)$ implies:  {\em $x$ divides $f(x)$ for all $x$}.  
First, observe that any finite set $L$ is $\times$-recognizable.
Indeed, letting $L\subseteq\{0,\ldots,a-1\}$ and considering the map $\varphi_{a,1}$,
we have $L=\varphi_{a,1}^{-1}(L)$. As Lemma \ref{l:varphi ak morphism semiring} insures that 
$\varphi_{a,1}$ is a $\times$-morphism, we conclude that $L$ is $\times$-recognizable. 

Let $L=\{f(a)\}$, then (by lemma~\ref{NX}) every set in $\+L_\times(L)$ is a set of divisors of $f(a)$. As $a\in f^{-1}(L)\in \+L_\times(L)$, we conclude that $a$ is a divisor of $f(a)$.

\smallskip

\if 32
{\bf Fact 2.} We next prove that $(i)$ implies:
{\em if $b$ divides $a$ then $f(a)=f(b)\times(a/b)^n$ for some $n\in\N$}. 
Let $q=a/b$ and argue by contradiction. Assume $f(a)\not\in M=\{f(b)q^n\mid n\in\N\}$. 
As $a$ divides $f(a)$, the quotient $f(a)/a$ is an integer, and the set of $k$'s such that
$f(a)/(a\times q^k)$ is an integer contains 0, hence is non empty: let $0\leq k_0$ be the largest  $k$ such that $f(a)/(a\times q^k)$ is an integer.  
As $a$ divides $f(a)$ so does $q=a/b$ hence the largest  $k$ such that $f(a)/( q^k)$ is an integer $k_1\geq 1$. Moreover, $k_1\geq k_0+1$. Indeed, clearly $k_0\leq k_1$, and $k_0$ cannot be equal to $k_1$, otherwise $f(a)/(a\times q^{k_1})=f(a)/(q\times b\times q^{k_1})=f(a)/( b\times q^{k_1+1})$ would be an integer, and a fortiori $f(a)/(q^{k_1+1})$ would be an integer, contradicting the definition of $k_1$.

Let $L=\{f(a)/( q^j)\mid 0\leq j\leq k_1\}$. 
Since we assume $f(a)\not\in M=\{f(b)q^n\mid n\in\N\}$,
inequality $k_1\geq1$ implies  $f(b)\not\in L$,
thus $b\not\in f^{-1}(L)$. 
 Now, $f(a)\in L$ and $a\in f^{-1}(L)$. 
By  Lemma~\ref{NX},  $\+L_\times(L)$ is generated by the $\gamma^{-1}(L)$ which are here of the form $L/i=\{f(a)/(i\times q^j)\mid 0\leq j\leq k_1\}$.
To derive a contradiction, it suffices to prove that for all $i$, 
if $a\in L/i$ then also $b\in L/i$.
 If $a\in L/i$ then for some $\ell\leq k_1$,  $a=f(a)/(i\times q^\ell)$. 
Thus $i=f(a)/(a\times q^\ell)$ is an integer and $\ell\leq k_0$ (by definiton of $k_0$). As $k_0<k_1$, $f(a)/ q^{\ell+1}\in L$ and 
thus $f(a)/(i\times q^{\ell+1})\in L/i$.
Replacing $i$ by its value, we have 
$f(a)/(i\times q^{\ell+1})=f(a)/\big(\,( f(a)/(a\times (a/b)^\ell))\times(a/b)^{\ell+1}\,\big)=b$. Whence the contradiction.
%
\fi

{\bf Fact 2.} We next prove that $(i)$ implies:
{\em if $b$ divides $a$ then $f(a)=f(b)\times(a/b)^n$ for some $n\in\N$}. 
The case $b=a$ being trivial we suppose $b<a$.
Let $q$ be the integer $q=a/b$ and set $L=\N\ \cap\ \{f(a)/( q^j)\mid j\in\N\}$.
This set is finite: if $f(a)/( q^j)\in\N$ then $q^j\leq f(a)$ hence $j\leq\lfloor\log(f(a)/\log(q)\rfloor$.
Being finite $L$ is recognizable
and condition (i) insures that $f^{-1}(L) \in \+L_\times(L)$.
Letting $j=0$ we see that $f(a)\in L$ hence $a\in f^{-1}(L)$.
We prove that $b$ is also in $f^{-1}(L)$.
Being in $\+L_\times(L)$, the set $f^{-1}(L)$ is of the form $\bigcup_{s\in S}\bigcap_{i\in S_s}L/i$
for some finite family $S$ of finite subsets $S_i\subset\N$, $i\in S$ 
(cf. Lemmas~\ref{l:normal LAL} \& \ref{NX})
hence $a\in\bigcap_{i\in S_s}L/i$ for some $s$.
In particular, to prove $b\in f^{-1}(L)$ it suffices to prove that
$a\in L/i\Rightarrow b\in L/i$ for all $i\geq1$. 
If $a\in L/i$ then for some $\ell$ we have $a=f(a)/(i\times q^\ell)$
Thus, $f(a)/(a\times q^\ell)=i$ is an integer. 
As $q=a/b$ we have $f(a)/(q^{\ell+1}) = b \times f(a)/(a\times q^\ell) = b\times i$.
Thus, $f(a)/(q^{\ell+1})$ is also an integer hence $f(a)/(q^{\ell+1}) \in L$.
Now, $f(a)/(i\times q^{\ell+1}) = (f(a)/(i\times q^\ell))\times (b/a) = a\times(b/a) = b$ is an integer
hence $b=f(a)/(i\times q^{\ell+1})\in L/i$. 
This proves that $b$ is in $f^{-1}(L)$ hence $f(b)\in L$, i.e. $f(b)=f(a)/(q^j)$ for some $j$
hence $f(a)=f(b) \times q^j = f(b)\times(a/b)^j$.
This finishes the proof of Fact 2.

\smallskip
We finally prove that $(i)$ implies: {\em $f$ is  of the form $f(x)=f(1)\times x^n$ for some fixed $n\in\N$}, i.e.,  $(i)\Rightarrow(ii)$.
We apply Fact 2  with various pairs $(a,b)$ and write $n(a,b)$ for the exponent 
such that $f(a)=f(b)\,(a/b)^{n(a,b)}$.
Let $u$ and $v$ be coprime.
\begin{align}\label{eq: xy 1}
f(uv) &= f(1)\, (uv)^{n(uv,1)} &\text{$(a,b)=(uv,1)$}
\\\label{eq: x 1}
f(u) &= f(1)\, x^{n(u,1)} &\text{$(a,b)=(u,1)$}
\\\label{eq: xy x}
f(uv) &= f(u)\, v^{n(uv,u)}  &\text{$(a,b)=(uv,u)$}
\\\label{eq: xy x 1}
f(uv)&= f(1)\, u^{n(u,1)} \, v^{n(uv,1)} &\text{(combine \eqref{eq: xy x} and \eqref{eq: x 1})}
\end{align}
As $u$ and $v$ are coprime, 
comparing the exponents of $u$ in \eqref{eq: xy 1} and \eqref{eq: xy x 1},
we conclude that $n(uv,1) = n(u,1)$, for all coprime $u,v$.
Exchanging the roles of $u$ and $u$ we get $n(uv,1) = n(v,1)$.
Thus, if $u,v$ are coprime then $n(u,1)=n(v,1)$.
Now, for every $x,y\geq1$ there exists $z$ coprime to both $x$ and $y$.
We then have $n(x,1)=n(z,1)$ and $n(y,1)=n(z,1)$.
Thus, $n(x,1)=n(y,1)$.
Let $k$ be the common value of the $n(x,1)$'s. Equation \eqref{eq: x 1}
insures that $f(x)=f(1)x^n$ for every $x\geq1$.



\bigskip

$(ii)\Longrightarrow (iii)$. Straightforward.

\bigskip

$(ii)\Longrightarrow (i)$. If $f$ is of the form given in $(ii)$, 
then $f^{-1}(L)=\sqrt[n]{L/f(1)}$ with $n\in\N$.
We proved (Theorem 2.2 in \cite{cgg14ipl})   that any lattice closed by division is also closed by $n$th root. As $\+L_{\times}(L)$ is the smallest lattice containing $L$ and closed under division, it  is also closed by $n$th root; this  implies that $\+L_{\times}$ is closed under $f^{-1}$.

\bigskip

$(iii)\Longrightarrow (ii)$.  For $x\in\N$, let $P(x)$ be the set of primes dividing $x$.
Recall that $\val(x,p)$ denotes the
highest exponent $n$ of $p$ such that $p^n$ divides $x$.

If $p$ is prime then the relation $\sim_p$ defined by
$x \sim_p y$  if and only if  $\val(x,p)=\val(y,p)$ is a congruence on $\langle  \N\setminus\!\{0\} ;\times\rangle$.
As $f$ is congruence preserving,
we see that $\val(x,p)=\val(y,p)$ implies $\val(f(x),p)=\val(f(y),p)$. 
In other words, the $p$-valuation of $f(x)$ depends only on that of $x$. 
Thus, there is function $\theta_p : \N \to \N$ such that $\val(f(x),p)=\theta_p(\val(x,p))$.

Consider now the $\times$-congruence defined by 
$x \sim_{p,q} y$  if and only if  $\val(x,p)+\val(x,q)=\val(y,p)+\val(y,q)$. 
Clearly $p\sim_{p,q} q$ and $p^k \sim_{p,q} p^{k-1} q$ for all $k\geq1$, 
hence $f(p) \sim_{p,q} f(q)$ and $f(p^k)\sim_{p,q} f(p^{k-1}q)$.
Thus, $\val(f(p),p)+\val(f(p),q) = \val(f(q),p)+\val(f(q),q)$
and $\val(f(p^k),p)+\val(f(p^k),q) = \val(f(p^{k-1}q),p)+\val(f(p^{k-1}q),q)$,
hence 
\begin{eqnarray}\label{eq:theta1 theta0}
\theta_p(1)+\theta_q(0)&=&\theta_p(0)+\theta_q(1)
\\\notag
\theta_p(k)&=&\theta_p(k-1)+\theta_q(1)-\theta_q(0) \text{\qquad for $k\geq1$}
\\\label{eq:thetap}
\text{this yields\quad}\theta_p(k) &=& \theta_p(0) + k\,n  \text{\qquad for all $k\in\N$}
\end{eqnarray}
where $n$ is the common value of the $\theta_p(1) - \theta_p(0)$'s for prime $p$,
a property insured by equation \eqref{eq:theta1 theta0}.

Let $F$ be the finite set of primes which divide $f(1)$.
This set is also the set of primes $p$ such that $\theta_p(0)\neq0$.
Using \eqref{eq:thetap} we see that 
$\val(f(x),p)=\theta_p(\val(x,p)) =  \theta_p(0) + \val(x,p)\,n$ for every prime $p$.
In particular, $\val(f(x),p)=0$ if $p\notin F\cup P(x)$.
Thus,
\begin{eqnarray}\notag
f(x) &=& \prod_{p\in F\cup P(x)}  p^{\theta_p(0) + \val(x,p)\,n} 
\\\label{eq:Px F}
&=& \prod_{p\in F\cup P(x)} p^{\theta_p(0)}
            \times \prod_{p\in P(x)\cup F} p^{\val(x,p)\,n} 
\\\label{eq:Px F bis}
&=& \prod_{p\in F} p^{\theta_p(0)} \times \prod_{p\in P(x)} p^{\val(x,p)\,n} 
\\\label{eq:fx f1 xn}
&=& f(1)\,x^n
\end{eqnarray}
where the passage from \eqref{eq:Px F} to \eqref{eq:Px F bis} is justified as follows:
\\\indent- for $p\notin P(x)$ we have $\val(x,p)=0$ hence $p^{\val(x,p)}=1$,
\\\indent- for $p\notin F$ we have $\theta_p(0)=0$ hence $p^{\theta_p(0)}=1$.
\\
Equation~\eqref{eq:fx f1 xn} is the wanted condition (ii).
\end{proof}

\section{ Lattices and  preservation of stable preorders or congruences}\label{s:genResFin}

We prove a variant of Theorem~\ref{thm:ipl1}  for general algebras, where  congruence preservation is replaced by the  stronger condition of 
stable preorder preservation (Section \ref{sec:spp}).
We  also extend  the stronger version of Theorem~\ref{thm:ipl1} with the weaker condition  {\bf (2)} to residually finite algebras (in a strong sense) admitting a group operation.  

\subsection{Stable preorder preservation and complete lattices}\label{sec:spp}
We prove Theorem \ref{th:LatticeGen}, a variant of Theorem~\ref{thm:ipl1}, where
stable preorder preservation, complete lattices and arbitrary subsets replace congruences, lattices
and recognizable subsets respectively.

\begin{lemma} \label{f^-1-dansTreillis} Let $\+A=\langle A;\Xi\rangle$.  
Given a subset $L\subseteq A$, if $f : A\to A$ preserves the syntactic preorder $\leq_L$
then we have
\begin{equation}\label{eq:f-1L}
f^{-1}(L) \ =\ \mathop{\bigcup}\limits_{a\in f^{-1}(L)}\left(\mathop{\bigcap}\limits_{\{\gamma\in  {\textit{\DUO}}(\+A)\mid \gamma(a)\in L\}}\gamma^{-1}(L)\right).
\end{equation}
Thus $f^{-1}(L) \in \+L_{\+A}^\infty(L)$.
\end{lemma}

\begin{proof} 
Let $I_a=\bigcap_{\{\gamma\in  {\textit{\DUO}}(\+A)\mid \gamma(a)\in L\}}\gamma^{-1}(L)$.
 Observe that $c\in I_a$ if and only if, for all $\gamma\in\DUO(\+A)$, we have
$\gamma(a)\in L\Rightarrow\gamma(c)\in L$.
Thus, by definition of the syntactic preorder $\leq_L$, we have
$I_a=\{c\mid c\leq_L a\}$.
Since $a\leq_L a$ we have $a\in I_a$ hence $f^{-1}(L)\subseteq\bigcup_{a\in f^{-1}(L)} I_a$.
Conversely, for every $c\in I_a$ we have $c\leq_L a$ and, as $f$ preserves $\leq_L$,  we also have
$f(c)\leq_L f(a)$ hence if $f(a)\in L$ then $f(c)\in L$.
Thus, $\bigcup_{a\in f^{-1}(L)} I_a \subseteq f^{-1}(L)$.
This proves \eqref{eq:f-1L}.
\end{proof}
\begin{remark} For $\N$ with operations $\suc$, +  or $\times$,  equation \eqref{eq:f-1L} can be simplified due to the associativity and commutativity of the operations.
\begin{itemize}
\item For  $\langle \N;\suc\rangle$ or $\langle \N;+\rangle$,  equation \eqref{eq:f-1L} reduces to $$f^{-1}(L) \ =\ \cup_{a\in f^{-1}(L)}\left({\cap}_{\{n\in L-a\}}L-n\right) \text{,\ where\ } L-a=\{x\mid x+a\in L\}.$$ 
\item For $\langle \N;\times\rangle$,  equation \eqref{eq:f-1L} becomes $$f^{-1}(L) \ =\ {\cup}_{a\in f^{-1}(L)}\left({\cap}_{\{n \in L/a\}}L/n \right), \text{\ where\ } L/n=\{x\mid nx\in L\}.$$ 
\item For the free monoid $\langle \Sigma^*;\cdot\rangle$, equation \eqref{eq:f-1L} becomes 
$$f^{-1}(L) \ =\ {\cup}_{a\in f^{-1}(L)}\left({\cap}_{\{(x,y) \mid xay \in L\}} x^{-1}Ly^{-1} \right), \text{\ where\ } x^{-1}Ly^{-1}=\{z\mid xzy\in L\}.$$
\end{itemize}
\end{remark}

\begin{theorem} \label{th:LatticeGen} 
[ Stable preorder preservation and lattices ]
Let $\+A=\langle A;\Xi\rangle$ be  an algebra  and $f$ be a mapping  $f\colon A\to A$. 

1) Conditions $(i)$ and $(ii)$ are equivalent
 and when they are satisfied, $f(a)\in \gen(a)$ for every $a\in A$.
\begin{itemize}
\item[(i)] 
$f$ is stable preorder preserving.
\item[(ii)]  
for every  subset  $L\subseteq A$, $f^{-1}(L)$ is in the complete lattice $\+L_{\+A}^\infty(L)$.
\end{itemize}

2) Conditions $(i')$ and $(ii')$ are equivalent
\begin{itemize}
\item[(i')] 
$f$ preserves the stable preorders which have finite index.
\item[(ii')]  
for every $\+A$-recognizable subset $L\subseteq A$, $f^{-1}(L)$ is in
the finite lattice $\+L_{\+A}(L)$.
\end{itemize}
\end{theorem}
\begin{proof}  
1)
$(i)\Longrightarrow (ii)$ :  follows from Lemma \ref{f^-1-dansTreillis}.

$(ii)\Longrightarrow (i)$: 
 Assume that $f^{-1}(L)$ is in $\+L_{\+A}^\infty(L)$ for all $L$.
Let $\leq$ be an $\+A$-stable preorder and $x\in A$ 
and consider the $\leq$-initial segment $L=\{z \mid z \leq f(x)\}$.
The assumed condition insures that $f^{-1}(L)$ is in $\+L_{\+A}^\infty(L)$.
Since $L$ is a $\leq$-initial segment so is also every set in $\+L_{\+A}^\infty(L)$
(by Lemma~\ref{bool:saturated}). In particular,  $f^{-1}(L)$ is a $\leq$-initial segment.
Observing that $x\in f^{-1}(L)$ (since $f(x)\leq f(x)$),
we deduce that if $y\leq x$ then $y\in f^{-1}(L)$
hence $f(y)\in L$ and therefore $f(y)\leq f(x)$.
This shows that $f$ preserves the stable preorder $\leq$.

In order to see that $f(a)\in \gen(a)$, it suffices to apply $(ii)$ to $L=\{f(a)\}$ :
 applying Lemma~\ref{f^-1-dansTreillis}, we have
$a\in f^{-1}(L)=  \cup_{i\in I} \big(\cap_{\gamma\in \Gamma_i}\gamma^{-1}(L)\big)$
hence, for some $\gamma\in {\textit{DUO}}(\+A)$, we have
$a\in \gamma^{-1}(L)=\gamma^{-1}(\{f(a)\})$ and therefore $\gamma(a)=f(a)$.

 2) $(i')\Longrightarrow (ii')$ follows from Lemma \ref{f^-1-dansTreillis}
and Lemma~\ref{l;recFinite}.

$(ii')\Longrightarrow (i')$.
Let $\leq$ be a stable preorder with finite index associated congruence.
Observe that the initial segment $L=\{z \mid z \leq f(a)\}$ is $\+A$-recognizable.
Indeed, $L$ is saturated for the congruence $\sim$ associated to $\leq$
and, as $\sim$ has finite index, we can apply Lemma~\ref{l:rec and congru}.
Thus, using again Lemma~\ref{l;recFinite}, we can follow the proof given for item 1).
 \end{proof}

The weak version
Theorem \ref{th:LatticeGen} differs in three ways from  Theorem~\ref{thm:ipl1}:
1) the set $L$ is arbitrary,
2) the lattice  $\+L_{\+A}^\infty(L)$ is complete,
3) the function is stable prorder preserving.
In the next subsections, we will get a result closer to  Theorem~\ref{thm:ipl1}  by assuming that
\begin{enumerate}
\item The algebra satisfies a convenient residual finiteness property (Definition~\ref{def:residually finite c sp}), in which case we  can restrict ourselves to considering finite index preorders, congruences, recognizable sets, and  lattices (instead of arbitrary congruences, sets and complete lattices).
\item The algebra has a group operation;  by Corollary \ref{coro:stable finite index preorder in group} 1), finite index congruence preservation is equivalent to finite index stable order preservation.
\end{enumerate}
%
\subsection{Congruence  preservation and (complete)  boolean algebras of subsets}\label{sec:cpbool}
Congruence preservation for arbitrary algebras can be characterized using boolean algebras instead of lattices.

\begin{lemma} \label{f^-1-dansBA} 
Let $\+A=\langle A;\Xi\rangle$, let $L$ be a subset of $A$ and $f:A\to A$. 
If $f$preserves the syntactic congruence $\sim_L$ then  $f^{-1}(L) \in \+B_{\+A}^\infty(L)$.
More precisely,
\begin{align}\label{eq:f-1L BA}
f^{-1}(L)  =
 \mathop{\bigcup}\limits_{a\in f^{-1}(L)}\ \mathop{\bigcap}\limits_{\gamma\in  {\textit{\DUO}}(\+A)}
 \text{If $\gamma(a)\in L$ then $\gamma^{-1}(L)$ else $A\setminus\gamma^{-1}(L)$}\text{)}
\end{align}
\end{lemma}

\begin{proof} 
Letting $I_a=\mathop{\bigcap}\limits_{\gamma\in  {\textit{\DUO}}(\+A)}
 \text{If $\gamma(a)\in L$ then $\gamma^{-1}(L)$ else $A\setminus\gamma^{-1}(L)$}\text{)}$,
we argue as in the proof of Lemma~\ref{f^-1-dansTreillis}.
Observe that $c\in I_a$ if and only if, for all $\gamma\in\DUO(\+A)$, we have
$\gamma(a)\in L\Leftrightarrow\gamma(c)\in L$.
Thus, by definition of the syntactic congruence $\sim_L$, we have
$I_a=\{c\mid c \sim_L a\}$.
Since $a\sim_L a$ we have $a\in I_a$ hence $f^{-1}(L)\subseteq\bigcup_{a\in f^{-1}(L)} I_a$.
Conversely, for every $c\in I_a$ we have $c\sim_L a$ and, as $f$ preserves $\sim_L$,  we also have
$f(c)\sim_L f(a)$ hence if $f(a)\in L$ then $f(c)\in L$.
Thus, $\bigcup_{a\in f^{-1}(L)} I_a \subseteq f^{-1}(L)$.
This proves \eqref{eq:f-1L BA}.
\end{proof}

\begin{theorem}\label{thm:cp bhull}
[Congruence preservation and Boolean algebras]
Let $\+A$ be an algebra
and $f\colon A \rightarrow A$.

1) Conditions $(i)$ and $(ii)$ are equivalent
\begin{itemize}
\item[(i)] 
$f$ is $\+A$-congruence preserving.
\item[(ii)]  
for every  subset  $L\subseteq A$, $f^{-1}(L)$ is in the complete Boolean algebra $\+B_{\+A}^\infty(L)$.
\end{itemize}

2) Conditions $(i')$ and $(ii')$ are equivalent
\begin{itemize}
\item[(i')] 
$f$ preserves the $\+A$-congruences which have finite index.
\item[(ii')]  
for every $\+A$-recognizable subset $L\subseteq A$, $f^{-1}(L)$ is in
the finite Boolean algebra $\+B_{\+A}(L)$.
\end{itemize}
\end{theorem}
\begin{proof}
1)  $(i)\Longrightarrow (ii)$ :  follows from Lemma \ref{f^-1-dansBA}.

$(ii)\Longrightarrow (i)$.
Assume that $f^{-1}(L)$ is in $\+B_{\+A}^\infty(L)$ for all $L$.
Let $\sim$ be an $\+A$-congruence and $x\in A$ 
and consider the $\sim$-class $L$ of $f(x)$.
The assumed condition insures that $f^{-1}(L)$ is in $\+B_{\+A}^\infty(L)$.
Since $L$ is $\sim$-saturated so is also every set in $\+B_{\+A}^\infty(L)$
(by Lemma~\ref{bool:saturated}). In particular,  $f^{-1}(L)$ is $\sim$-saturated.
Observing that $x\in f^{-1}(L)$ (since $L$ is the class of $f(x)$),
we deduce that if $x\sim y$ then $y\in f^{-1}(L)$
hence $f(y)\in L$ and therefore $f(y)\sim f(x)$
(again because $L$ is the class of $f(x)$).
This shows that $f$ preserves the congruence $\sim$. 
\\
2)
$(i')\Longrightarrow (ii')$.
If $L$ is recognizable then $\sim_L$ has finite index and hypothesis (i') insures that
$f$ preserves $\sim_L$. Applying Lemma \ref{f^-1-dansBA} and Lemma~\ref{l;recFinite}
we see that $f^{-1}(L)\in\+B_\+A(L)$.

$(ii')\Longrightarrow (i')$.
Let $\sim$ be a congruence with finite index, let $x\in A$ 
and consider the $\sim$-class $L$ of $f(x)$.
By  Lemma~\ref{l:rec and congru} the set $L$ is $\+A$-recognizable.
The assumed condition (ii') insures that $f^{-1}(L)$ is in $\+B_{\+A}(L)$.
Since $L$ is $\sim$-saturated so is also every set in $\+B_{\+A}(L)$
(by Lemma~\ref{bool:saturated}). In particular,  $f^{-1}(L)$ is $\sim$-saturated.
We conclude as in the proof given for item 1).
\end{proof}

%
\subsection{Residually finite algebras, recognizability and lattices}\label{s:43}

In the vein of \cite{KaarliPixley} (page 102), we define   notions of residual finiteness stronger than the classical  ones for congruences,
 preorders and algebras  tailored to fit in our framework.

\begin{definition}\label{def:residually finite c sp}1)  A congruence on an algebra $\+A$ is {\em c-residually finite}
if it is the intersection of a family of finite index congruences.

2) A stable preorder on an algebra $\+A$ is {\em sp-residually finite}
if it is the intersection of a family of stable preorders
all of which have finite index associated congruences.

3) An algebra $\+A$ is said to be {\em c-residually finite} if all congruences on  $\+A$ 
are c-residually finite.
$\+A$ is said to be {\em sp-residually finite}
if all stable preorders on  $\+A$ are {\em sp-residually finite}.
\end{definition}
\begin{remark}
The usual notion of residually finite group, ring or module requires that morphisms into finite algebras 
separate points, i.e., if $x\neq y$ there exists a morphism $\varphi$ into a finite algebra 
such that $\varphi(x)\neq\varphi(y)$.
This notion is equivalent to the c-residual finiteness of a single congruence, 
the trivial identity congruence,
hence it is weaker than that of c-residually finite algebra.
\end{remark}

 Every congruence being a preorder, 
Definition~\ref{def:residually finite c sp} gives a priori two notions of
residual finiteness for a congruence. In fact, both notions are proven to coincide below.

\begin{lemma}\label{l:p residually finite implies c}
1) If a stable preorder $\preceq$ is sp-residually finite then its associated congruence $\sim$
is c-residually finite.

2) A congruence is c-residually finite if and only if, as a preorder, it is sp-residually finite.

3) A sp-residually finite algebra is also c-residually finite.
\end{lemma}

\begin{proof}
1) Let $(\preceq_i)_{i\in I}$ be a family of stable preorders such that $\preceq=\bigcap_{i\in I} \preceq_i$.
Let $\sim_i$ be the congruence associated to $\preceq_i$.
We show that $\sim=\bigcap_{i\in I} \sim_i$.
As $\sim_i$ is included in $\preceq_i$ we have
$(\bigcap_{i\in I} \sim_i) \subseteq (\bigcap_{i\in I} \preceq_i) = \preceq$.
As $\bigcap_{i\in I} \sim_i$ is a congruence, the last inclusion yields
$(\bigcap_{i\in I} \sim_i) \subseteq \sim$.
The inclusion of the congruence $\sim$ in the preorder $\preceq$ 
together with the inclusion $\preceq\subseteq\preceq_i$  imply 
the inclusion $\sim \subseteq \preceq_i$.
As $\sim$ is a congruence, this last inclusion yields $\sim \subseteq \sim_i$.
Thus, $\sim \subseteq (\bigcap_{i\in I} \sim_i)$.

2) If a congruence $\sim$ is c-residually finite and 
$\sim=\bigcap_{i\in I}\sim_i$ then the congruences $\sim_i$'s, being also preorders,
witness that $\sim$ is sp-residually finite. 
Conversely, applying 1) to a congruence $\sim$, we see that if $\sim$ is sp-residually finite
then it is also c-residually finite.

3) Trivial consequence of 2).
\end{proof}

%
%
%

\begin{example}\label{ex:PrimeSupport}
1) Integer (semi)-groups $\langle \N;+\rangle$ and $\langle \Z;+\rangle$ 
are c-residually finite as every non trivial congruence is of finite index (cf.  Lemmata \ref{folk0} and \ref{lem:CPsurA}) and the identity congruence is the intersection of all non trivial congruences. 
They are also sp-residually finite. Let  for instance, $\preceq$ be a stable preorder on $\langle \N;+\rangle$. For $i\in\N$  the relation $\preceq_i$ defined  by $x\preceq_i y$ if and only if 

$\bullet$ either $x<i$ or $y<i$ and $x\preceq y$,

$\bullet$ or $x\geq i$ and $y\geq i$,

\noindent is a stable preorder.  The congruence $\sim_i$ associated to $\preceq_i$ identifies all elements larger than $i$, thus   $\sim_i$  has finite index. It is easy to see that $\preceq \ = \ \cap_{i\in\N} \preceq_i$, hence $\preceq$ is residually finite and $\langle \N;+\rangle$ is sp-residually finite.

2) Contrary to the previous example, for $k\geq 2$,
the algebra $\langle \N^k;+\rangle$ admits c-residually finite congruences having infinite index, e.g.,
the congruence $\vec{x}\sim \vec{y} \Leftrightarrow x_1= y_1$. It is residually finite because $x_1= y_1$ if and only if $x_1\equiv y_1\pmod n$, for all $n$.

3) On the algebra of integers with multiplication 
$\langle \N;\times\rangle$, there exist residually finite non trivial congruences with  infinite index.  For instance consider on $\langle \N;\times\rangle$  the congruence $x\sim  y$ if and only if $x$ and $y$ have the same set of primes divisors:  $\sim $ does not have a finite index.  For each prime number $p$ let $\sim _p$ be the congruence $x\sim_p y$ if and only if $p$ divides both $x,y$ or neither of them. Each $\sim _p$ has finite index 2 and $\sim\  =\ \cap_{p\in \PP} \sim _p$.
\end{example}
\begin{lemma}\label{l:residually finite utile}
Let $\+A = \langle A;\Xi\rangle$ be an algebra and $f:A\to A$.

1) If $\+A$ is sp-residually  finite then $f$ is stable preorder preserving if and only if
$f$ preserves all stable preorders having finite index associated congruences

2) If $\+A$ is c-residually finite then $f$ is congruence preserving if and only if
$f$ preserves all finite index congruences.
\end{lemma}

\begin{proof}
1) Let $\leq$ be a stable preorder. 
The  hypothesis of sp-residual finiteness of $\+A$ insures that $\preceq$ is sp-residually  finite:
there exists a family of stable preorders $(\preceq_i)_{i\in I}$  
with associated congruences having finite indexes, such that $\preceq=\cap_{i\in I} \preceq_i$.
Thus, $a\preceq b$ if and only if, for all $i\in I$, $a\preceq_i b$.  
The hypothesis insures that $f$ preserves the $\preceq_i$'s hence 
$f(a)\preceq_i f(b)$ for all $i\in I$. This yields $f(a)\preceq_i f(b)$.

The proof of 2) is similar.
\end{proof}

The next theorem  improves item 2 of Theorem~\ref{th:LatticeGen} for 
sp-residually finite algebras.

\begin{theorem}\label{th:RecLatGeneralResiduallyFinite} 
Let $\+A = \langle A;\Xi\rangle$ be an algebra.

1) If $f\colon  A\to A$ is stable preorder preserving then 
$f^{-1}(L)$ is in the lattice $\+L_{\+A}(L)$ for every $\+A$-recognizable $L\subseteq A$.

2) Assume $\+A$ is  a sp-residually finite algebra.
If $f\colon  A\to A$ is such that 
$f^{-1}(L)$ is in the lattice $\+L_{\+A}(L)$ for every $\+A$-recognizable $L\subseteq A$
then $f$ is stable preorder preserving. 
\end{theorem}

\begin{proof}
1)  Apply implication $(i)\Rightarrow(ii)$ of Theorem~\ref{th:LatticeGen}
and then Lemma~\ref{l;recFinite}.

2)  Applying implication $(ii')\Rightarrow(i')$ of Theorem~\ref{th:LatticeGen}, 
we already know that $f$ preserves stable preorders with finite index associated congruences.
To conclude apply Lemma~\ref{l:residually finite utile}.
\end{proof}

\begin{remark}  If $\+A $ is sp-residually finite, Theorem \ref{th:RecLatGeneralResiduallyFinite} states the equivalence:
$f\colon  A\to A$ is stable preorder preserving if and only if  
$f^{-1}(L)$ is in the lattice $\+L_{\+A}(L)$ for every $\+A$-recognizable $L\subseteq A$.
An instance of Theorem \ref{th:RecLatGeneralResiduallyFinite} for the sp-residually finite algebra $\+N=\langle \N ;  +, \times\rangle$
is the characterization of congruence preservation given in Theorem \ref{thm:ipl1}. 
Indeed, Proposition \ref{prop:miracle}, 
shows that on $\+N$ a function $f$ is stable preorder preserving 
if and only if it is monotone non  decreasing and congruence preserving.  Theorem \ref{thm:CPsurN} shows that a
 non constant congruence preserving $f$ satisfies $f(x)\geq x$ for all $x$.  Thus on $\+N$, Theorem \ref{thm:ipl1}  becomes a consequence of Theorem \ref{th:RecLatGeneralResiduallyFinite}.
\end{remark}

The next theorem  improves item 2 of Theorem~\ref{thm:cp bhull} for 
c-residually finite algebras.

\begin{theorem}\label{th:RecLatGeneralResiduallyFinite c} 
Let $\+A = \langle A;\Xi\rangle$ be an algebra.

1) If $f\colon  A\to A$ is congruence preserving then 
$f^{-1}(L)$ is in the Boolean algebra $\+B_{\+A}(L)$ for every $\+A$-recognizable $L\subseteq A$.

2) Assume $\+A$ is  a c-residually finite algebra.
If $f\colon  A\to A$ is such that 
$f^{-1}(L)$ is in the Boolean algebra $\+B_{\+A}(L)$ for every $\+A$-recognizable $L\subseteq A$
then $f$ is congruence preserving. 
\end{theorem}

\begin{proof}
1)  Apply implication $(i)\Rightarrow(ii)$ of Theorem~\ref{thm:cp bhull}
and then Lemma~\ref{l;recFinite}.

2) Applying implication $(ii')\Rightarrow(i')$ of Theorem~\ref{thm:cp bhull}, 
we already know that $f$ preserves finite index congruences.
To conclude apply Lemma~\ref{l:residually finite utile}.
\end{proof}

\subsection{Stable preorders become congruences when there is a group operation}
\label{ss:when there is a group operation}

 In some frameworks the requirement that a preorder be stable is quite a strong requirement 
 as shown by Corollary~\ref{coro:stable finite index preorder in group} below.
Recall first the notion of cancellability.
 \begin{definition}\label{def:cancellative}
1) A semigroup $S$ is said to be cancellable if  $xz=yz$ implies $x=y$ and $zx=zy$ implies $x=y$.

2)  A stable preorder  $\preceq$  on $S$  is said to be cancellable
if $xz\preceq yz$ implies $x\preceq y$ and $zx\preceq zy$ implies $x\preceq y$.
\end{definition}

\begin{lemma}\label{l:stable order on finite group}
 The only stable order of a finite group $G$ is the identity relation.
\end{lemma} 
\begin{proof}
Assuming $x\preceq y$ we prove $x=y$.
Let $e$ be the unit of $G$.
Stability under left product by $x^{-1}$ and $(x^{-1}\,y)^n$
successively yield $e\preceq x^{-1}\,y$
and then $(x^{-1}\,y)^n \preceq (x^{-1}\,y)^{n+1}$ for all $n\geq1$.
By transitivity,  $e \preceq x^{-1}\,y \preceq (x^{-1}\,y)^n$.
As the group is finite there exists $k$ and $n\geq 1$ such that $(x^{-1}\,y)^{k+n}=(x^{-1}\,y)^k$
hence $(x^{-1}\,y)^n=e$; 
thus $e \preceq x^{-1}\,y \preceq e$ and by antisymmetry $e = x^{-1}\,y$ and $x=y$.
\end{proof} 
\begin{lemma}\label{l:finite cancellative semigroup}
 Any  finite cancellable semigroup is a group.
\end{lemma} 
\begin{proof}
Cancellability insures that, for every $a\in S$, 
the maps $x\mapsto ax$ and $x\mapsto xa$ are injective 
hence are bijections $S\to S$ because $S$ is finite. 
In particular, for all $a\in S$ there exist $e'_a, e''_a$ such that $a e'_a=a$ and $e''_a a=a$.
For $a,b\in S$ we then have $ae'_a b = ab = a e''_b b$ and by cancellability $e'_a=e''_b$,
proving that $e'_a$ and $e''_a$ do not depend on $a$ and are equal.
Thus, the common value $e$ of the $e'_a$'s and $e''_a$'s is a unit of $S$.
Also, for $a\in S$ there exists $a', a''$ such that $aa'=a''a=e$.
Then Then $a' = ea' = (a''a)a' = a''(aa') = a''e = a''$, proving that $a'$ is an inverse of $a$.
\end{proof}

\begin{proposition}\label{p:stable cancellable preorder with finite index}
Let $\+A = \langle A;\Xi\rangle$ be an algebra such that $\Xi$ contains a semigroup operation.
Every cancellable stable preorder $\preceq$ of $\+A$ 
 such that the associated congruence $\sim$ has finite index 
is equal to its associated congruence $\sim$.
\end{proposition}
\begin{proof}
The semigroup operation on $\+A$ induces a semigroup operation
on the quotient algebra $\+G = \+A/\!\!\sim$ with carrier set $A/\sim$.
The cancellability property of the congruence $\sim$ yields 
the cancellability property of the semigroup operation on $G$.
Indeed, suppose $X,Y,Z\in G$ satisfy $XZ=YZ$ 
and let $x,y,z\in A$ be representatives of the classes $X,Y,Z$. 
Then we have $xz\sim yz$ and cancellability yields $x\sim y$
hence $X=Y$. Idem if $ZX=ZY$.
As $\sim$ has finite index, $G$ is finite and Lemma~\ref{l:finite cancellative semigroup}
insures that the semigroup operation on $G$ is a group operation.
As $\sim$ is the congruence associated to the stable preorder $\preceq$,
it induces a quotient stable order $\preceq/\!\!\sim$ on the finite quotient algebra $\+G$.
As $\+G$ is an expansion of a finite group,
Lemma~\ref{l:stable order on finite group} insures that $\preceq/\!\!\sim$ is the identity relation on $G$
hence $\preceq$ coincides with $\sim$ and is therefore a congruence on $\+A$.
\end{proof}
\begin{corollary}\label{coro:stable finite index preorder in group}
Let $\+A = \langle A;\Xi\rangle$ be an algebra such that $\Xi$ contains a group operation.

 1) Every stable preorder with finite index associated congruence is a congruence.

2) Every sp-residually finite preorder is a c-residually finite congruence.

3) If $\+A$ is c-residually finite then it is also sp-residually finite.
\end{corollary}
\begin{proof}
Observe that stability of $\preceq$ implies its cancellability:
if $xz \preceq yz$ then $xzz^{-1} \preceq yzz^{-1}$ hence $x \preceq y$.
Idem if $zx \preceq zy$.
Then apply Proposition~\ref{p:stable cancellable preorder with finite index}.
\end{proof}
\begin{corollary}\label{coro:syntactic preorder recognizable is congruence in group}
Let $\+A = \langle A;\Xi\rangle$ be an algebra such that $\Xi$ contains a group operation.
If $L\subseteq A$ is recognizable 
then its syntactic preorder is equal to its syntactic congruence. 
\end{corollary}
\begin{proof}
Recall that the syntactic congruence of a recognizable set has finite index
and apply Corollary~\ref{coro:stable finite index preorder in group}.
\end{proof}

\subsection{Congruence preservation when there is a group operation}
\label{ss:stable preorders recognizability lattices}
When there is a group operation in a c-residually finite algebra, Theorem \ref{th:LatticeGen} can be given a more interesting variant form  by replacing stable preorder preservation by congruence preservation, complete lattices by lattices and subsets by recognizable subsets.

 In case the algebra has a group operation and is c-residually finite, 
there is a collapse of the diverse conditions involving 
congruence preservation, stable preorder preservation,
inverse images of recognizable sets, lattices and Boolean algebras.

\begin{theorem}\label{th:RecLatGeneralGroup} 
Let $\+A = \langle A;\Xi\rangle$ be a  c-residually finite algebra
such that $\Xi$ contains a group operation.  
Let $f\colon  A\to A$. The following conditions are equivalent
\begin{itemize}
\item[(i)]$f$ is stable preorder preserving,

\item[(ii)]$f$ preserves stable preorder having finite index associated congruences,

\item[(iii)] 
$f$ is congruence preserving,

\item[(iv)] 
$f$ is preserves finite index congruences,

\item[(v)]  $f^{-1}(L)$ is in the lattice $\+L_{\+A}(L)$ for every $\+A$-recognizable $L\subseteq A$.

\item[(vi)]  $f^{-1}(L)$ is in the Boolean algebra $\+B_{\+A}(L)$ for every $\+A$-recognizable 
$L\subseteq A$.
\end{itemize}

\end{theorem}

\begin{proof}

Item 3 of Corollary~\ref{coro:stable finite index preorder in group} insures that
$\+A$ is also sp-residually  finite.
Thus, equivalences $(i)\Leftrightarrow(ii)$ and $(iii)\Leftrightarrow(iv)$ are given by
Lemma~\ref{l:residually finite utile}.

Item 1) of Corollary~\ref{coro:stable finite index preorder in group} yields
the equivalence $(ii)\Leftrightarrow(iv)$.

Equivalences $(i)\Leftrightarrow(v)$ and $(iii)\Leftrightarrow(vi)$ are given by
Theorems~\ref{th:RecLatGeneralResiduallyFinite} and \ref{th:RecLatGeneralResiduallyFinite c}.

Summing up, we have 
$(v)\Leftrightarrow(i)\Leftrightarrow(ii)\Leftrightarrow(iv)\Leftrightarrow(iii)\Leftrightarrow(vi)$.
\end{proof}

In Section \ref{ss:Z} we will apply  Theorem \ref{th:RecLatGeneralGroup} to $\langle\Z,+\rangle$.

\section{ Case of integers $\langle  \Z ;+\rangle$ and  $\langle  \Z ;+,\times\rangle$} 
\label{ss:Z}

 In this section we look for an  extension of Theorem~\ref{thm:ipl}  to functions $\Z\to\Z$, for the structures $\+Z=\langle  \Z ;+\rangle$ and $\+Z'=\langle  \Z ;+,\times\rangle$.  
 
\subsection{Congruences on $\langle\Z;+\rangle$ and $\langle\Z;+,\times\rangle$}

Recall that the congruences of $\langle\Z;+\rangle$ are the equality relation and the modular
congruences $x\equiv y\pmod k$ for $k\geq1$.
These $\langle\Z;+\rangle$-congruences are also $\langle\Z;+,\times\rangle$-congruences.
Thus, applying item 2 of Definition~\ref{def:finite monogenic monoids},
we have the following $\Z$ avatar of Corollary~\ref{+CPimpliesXCPN} about $\N$.

\begin{lemma}\label{+CPimpliesXCP} 
The two structures 
$\langle \Z;+,\times\rangle$ and $\langle \Z;+\rangle$
yield the same notions of congruence (namely, equality and the modular congruences),
congruence preserving function $\Z\to \Z$, morphism $\Z\to\Z$ \
and recognizable subset of $\Z$.
\end{lemma}
Hence on  $\Z$, the study of congruence-preservation and recognizability w.r.t. to the signature + supersedes the study w.r.t. the signature $+,\times$. However congruence-preservation and recognizability w.r.t. to the signature $\times$ yield no consequence for congruence-preservation and recognizability w.r.t. to the signature +
as  not every $\langle  \Z ;\times\rangle$-morphism (resp. congruence, recognizable set) is a $\langle  \Z ;+\rangle$-morphism (resp. congruence, recognizable set).

In general, congruences are kernels of morphisms into possibly infinite algebras.
However, for the ring of integers (cf. Lemma \ref{lem:CPsurA}), 
non trivial congruences coincide with kernels of  morphisms onto finite structures, 
exactly as for with the semiring of natural numbers.
This allows to consider only congruences having finite index.

\subsection{Congruence preservation on principal commutative rings}
More generally than $\langle  \Z ;+,\times\rangle$, we  characterize congruence preservation for commutative principal rings with the signature $\Xi=\{+,\times\}$.
\begin{lemma}\label{lem:CPsurA}
If $\+A$ is a principal commutative ring (i.e., every ideal is principal), then\\
(i) any congruence $\sim$ is of the form
$\sim_k =\{(u,v)\mid u-v\in k\,A\}$ for some $k\in A$.\\
(ii) a function $f:A\to A$ is congruence preserving if and only if it satisfies
\begin{equation}\label{eq:f freezing cp}
\text{\begin{tabular}{l}
$x-y$ divides $f(x)-f(y)$ for all $x,y\in X$\end{tabular}}
\end{equation}
\end{lemma}
\begin{proof} 
 The hypothesis that $\+A$ is principal yields condition $(i)$.

$(ii)$ 
Assume $f$ is congruence preserving. For $x,y\in A$,
 let $\sim =\{(u,v)\mid u-v\in (x-y)\,A\}$ be the congruence 
generated by the ideal $(x-y) A$.
Since $x\sim y$, congruence preservation insures that $f(x)\sim f(y)$
hence $x-y$ divides $f(x)-f(y)$.

Conversely, assume \eqref{eq:f freezing cp} holds and let $\sim$ be a congruence, which is of the form $\sim_k$ because of principality.
If  $x\sim_k y$ then $k$ divides $x-y$ and, by transitivity of divisibility,
\eqref{eq:f freezing cp} insures that $k$ divides $f(x)-f(y)$, i.e., $f(x)\sim_k f(y)$.
\end{proof}
\begin{remark}
 Otherwise stated, for principal commutative rings, congruence preservation is equivalent  to condition $(2) (i)$ of Theorem \ref{thm:CPsurN} and  the over-linearity condition $(2) (ii)$ 
 of Theorem \ref{thm:CPsurN} is not needed.
\end{remark}

\subsection{Recognizability in $ \langle  \Z ;+\rangle$ and $ \langle  \Z ;+,\times\rangle$}
%

Recall first a folk Proposition.
\begin{proposition}\label{p:rec Z}
Let $X\subseteq\Z$. The following conditions are equivalent:
\begin{enumerate}
\item[(i)]
$X$ is $ \langle  \Z ;+\rangle$-recognizable,
\item[(ii)]
$X$ is $ \langle  \Z ;+,\times\rangle$-recognizable,
\item[(iii)]
$X$ is of the form $X=F+k\Z$ with $k\in\N\setminus\{0\}$ and $F\subseteq\{0,\ldots,k-1\}$.
\end{enumerate}
\end{proposition}
\begin{proof}
$(i)\Rightarrow(ii)$.
Assume $(i)$ and let $\varphi: \langle  \Z ;+\rangle\to \langle M;\oplus\rangle$ be a surjective morphism
where $M$ has $k$ elements.
Since  $ \langle  \Z ;+\rangle$ is a  monogeneous group so is $\langle M;\oplus\rangle$
which is therefore isomorphic to $\Z/k\Z$.
Also, the morphism $\varphi$ is the modular projection $x\mapsto x\pmod k$.
To conclude that $(ii)$ is true, recall that the modular projection is also a ring morphism
$ \langle  \Z ;+,\times\rangle\to \langle\Z/k\Z;+,\times\rangle$.
\smallskip\\
$(ii)\Rightarrow(iii)$. 
Assume $(ii)$ and let $X=\varphi^{-1}(F)$ with $M$ finite,  $F\subseteq M$, $\varphi$ a surjective morphism $(\Z,+,\times)\to\langle M;\oplus,\otimes\rangle$.
We know (from the proof of $(i)\Rightarrow(ii)$) that, up to an isomorphism,
$\langle M;\oplus,\otimes\rangle$ is the ring $\Z/k\Z$ for some $k\in\N\setminus\{0\}$
and $\varphi:\Z\to\Z/k\Z$ is the modular projection.
Thus, $X=\varphi^{-1}(F)=F+k\Z$.
\smallskip\\
$(iii)\Rightarrow(i)$. Observe that $X=F+k\Z=\varphi^{-1}(F)$ 
where $\varphi:\Z\to\Z/k\Z$ is the modular projection.
\end{proof}

In $ \langle  \N ;+\rangle$ recognizable subsets cioncide with   regular subsets. In $ \langle  \Z ;+\rangle$ this is no longer true. A subset $L\subseteq\Z$ is regular if it is of the form
$L=L^+\cup (-L^-)$ where $L^+,L^-$ are regular subsets of $\N$,
i.e., $L=-(d+S+d\N)\cup F\cup(d+R+d\N)$ with
$d\geq1$,
$R,S\subseteq\{x\mid0\leq x<d\}$,
$F\subseteq\{x\mid-d<x<d\}$
(possibly empty). See \cite{benois}.

\begin{corollary} \label{cor:Z residually finite}
Both $\langle  \Z ;+\rangle$ and $\langle  \Z ;+,\times \rangle$ 
 are c-residually finite.
\end{corollary}
\begin{remark}
The only finite $\langle  \Z ;+\rangle$-recognizable set is the emptyset.
\end{remark}
%

\subsection{Lattices in $ \langle  \Z ;+\rangle$ and $ \langle  \Z ;+,\times\rangle$}

\begin{definition} Let  $\+R$ be a unit (semi)ring $\langle R;+,\times\rangle$, with $0,1$
as distinct identities for $+$ and $\times$. 

$\bullet\ $ $\+L_{\langle R;+\rangle}(L)$ is the smallest sublattice of  $\+P(R)$ containing $L$
and closed under $(x\mapsto x+a)^{-1}$ for all $a\in R$
(closed under {\em translations}).

 $\bullet\ $  $\+L_{\langle R;+,\times\rangle}(L)$ (resp. $\+L^\infty_{\langle R;+,\times\rangle}(L)$)

 is the smallest (resp. complete) sublattice of  $\+P(R)$ containing $L$ and closed under both translations and {\em divisions} 
(i.e., $(x\mapsto ax)^{-1}$ for all $a\in R$).
\end{definition}

By Lemma \ref{+CPimpliesXCP}, congruence preservation (resp. recognizability) w.r.t. $\langle  \Z ;+\rangle$ and $\langle  \Z ;+,\times\rangle$ are equivalent. The next Lemma shows that 
 this goes on with lattices, i.e., 
the lattices $\+L_{\langle  \Z ;+\rangle}(L)$ and $\+L_{\langle  \Z ;+,\times\rangle}(L)$ coincide for any recognizable $L$.
\begin{lemma}[Characterization of the lattice generated by a recognizable subset and closed by translations]\label{l:lattice L Z}
Let $L\subseteq\Z$ be a nontrivial (i.e., different from $\Z$ and  $\emptyset$) recognizable subset of   $\langle  \Z ;+\rangle$.
Let $k\geq1$ be smallest such that $L=F+k\Z$ with $F\subset\{0,\ldots,k-1\}$.
Then 
\[\+L_{\langle  \Z ;+\rangle}(L)=
\+L_{\langle  \Z ;+,\times\rangle}(L)=\{G+k\Z\mid G\subseteq\{0,\ldots,k-1\}\}.\]
\noindent 
 If $L=\emptyset$ or $L=\Z$ (i.e. $k=0$ or $k=1$)
then these three lattices coincide with $\{L\}$.
\end{lemma}
%


\begin{proof} 
We first prove that 
$\+L_{\langle  \Z ;+\rangle}(L)\supseteq\{G+k\Z\mid k\in\N{\text{ and }}G\subseteq\{0,\ldots,k-1\}\}$.
Let $k\geq1$ and $F$ be a nonempty subset
of $\{0,\ldots,k-1\}$, i.e., $F=\{z_1,\ldots,z_n\}$ with $n\geq1$ and $0\leq z_1<\ldots<z_n<k$.
For $i=1,\ldots,n$ consider the set 
$A_i=\{z_j-z_i\pmod k\mid j=1,\ldots,n\}\subseteq\{0,\ldots,k-1\}$.
Clearly, $0$ is in each $A_i$.
We claim that $\bigcap_{i=1}^{i=n}A_i=\{0\}$.
 If $k=1$ this is clear since then every $A_i$ is $\{0\}$.
We now assume $k\geq2$ and $\bigcap_{i=1}^{i=n}A_i \neq \{0\}$. 
Let $a\in\bigcap_{i=1}^{i=n}A_i$ with $a\neq 0$.
Then, there exists $\theta:\{1,\ldots,n\}\to\{1,\ldots,n\}$ such that $z_{\theta(i)}-z_i\equiv a\pmod k$ for $i=1,\ldots,n$;
hence $z_i+a\in z_{\theta(i)}+k\,\Z\subseteq F+k\,\Z$.
Since $F=\{z_1,\ldots,z_n\}$ we get $F+a\subseteq F+k\Z$, and by induction, for all $n\in\N$, $F+na\subseteq F+k\Z$. Indeed $F+na=F+(n-1)a+a\subseteq  F+k\Z +a=F+a+k\Z\subseteq F+k\Z+k\Z=F+k\Z$.
Let $d=\gcd(k,a)$. 
Using B\'ezout identity, there are  $p,q\in\Z$, $p>0>q$  such that 
$p a + q k=d$ and $p',q'\in\Z$, $p'<0<q'$  such that 
$p'a + q ' k=d$. Inclusion $F+p a\subseteq F+k\Z$ yields 
$F+d=F+p a+q k\subseteq F+k\Z$ hence (again by induction)  $F+d\N\subseteq F+k\Z$.
In the same way, using $p',q'$ we get $F-d\N\subseteq F+k\Z$.
Thus, $F+d\Z\subseteq F+k\Z$.
As $d$ divides $k$ we also have $F+k\Z\subseteq F+d\Z$.
Thus, $F+d\Z=F+k\Z=L$ and since $0<d<k$ (recall $0<a<k$ and $d$ divides $a$)
this contradicts the minimality of $k$.

Equality $\bigcap_{i=1}^{i=n}A_i=\{0\}$ implies that
\begin{eqnarray*}
\bigcap_{i=1}^{i=n} L-z_i &=& \bigcap_{i=1}^{i=n}\big( \{z_j-z_i\mid j=1,\ldots,n\}+k\Z\big)
= \bigcap_{i=1}^{i=n} \big(\{A_i\}+k\Z\big)\\
&=& \big(\bigcap_{i=1}^{i=n} A_i\big)+k\Z=
\{0\}+k\Z\ =\ k\Z.
\end{eqnarray*}
Thus, for all $b\in \{0,\ldots,k-1\}$, $b+k\Z= \bigcap_{i=1}^{i=n} \big(L-z_i+b\big)$ 
 belongs to $\+L_{\langle  \Z ;+\rangle}(L)$. 
All finite unions and intersections of such $b+k\Z$ also belongs to  $\+L_{\langle  \Z ;+\rangle}(L)$ proving that for all $G\subseteq\{0,\ldots,k-1\}$, 
 $\{G+k\Z\}$ is in  $\+L_{\langle  \Z ;+\rangle}(L)$. 

The converse inclusion is straightforward. 

\smallskip

Finally, as $\+L_{\langle  \Z ;+,\times\rangle}(L)$ is the smallest lattice closed by translations and divisions containing $L$, it suffices to prove that $\{G+k\Z\mid G\subseteq\{0,\ldots,k-1\}\}$ is closed by divisions to conclude. Recall that $L=F+k\Z$ with $ F\subseteq\{0,\ldots,k-1\}$. For $d\in\Z$ 
 $L/d=\{b\mid bd\in L\}$.
Set $G=L/d\cap  \{0,\ldots,k-1\}$, we show that $L/d = G+k\Z$. Clearly $(G+k\Z)d\subseteq dG+dk\Z\subseteq L+k\Z=L$, hence $G+k\Z\subseteq L/d$. 
Conversely, if $b\in L/d$, then 
$bd=f+kz\in L,\ f\in F$; letting $a=b\pmod k\in \{0,\ldots,k-1\}$, we have $ad=(b+kz')d=f+k(z+dz')\in L$ hence $a\in L/d\cap  \{0,\ldots,k-1\}$. 
Thus $b\in G+k\Z$, this shows that $L/d \subseteq G+k\Z$.
Hence  $L/d = G+k\Z$.

\smallskip\noindent Finally, the last assertion about the cases $L=\emptyset$ and $L=\Z$ is straightforward.
\end{proof}
\begin{remark} Note the following immediate consequence of Lemma \ref{l:lattice L Z}.
For every $L\neq\emptyset$ we have $\Z\in\+L_{\langle  \Z ;+\rangle}(L)$ as $\Z=\{0,\ldots,k-1\}+k\Z\}$. This is
different from the case of $\N$ where $\N$ does not necessarily belong to  $\+L_\N(L)$,
for instance when $L$ is finite (hence recognizable in $\langle\N;+\rangle$) all sets in $\+L_{\langle  \N ;+\rangle}(L)$  are finite.
\end{remark}
%
\subsection{Characterizing congruence preserving functions on $\langle  \Z ;+\rangle$}\label{s:55}
Theorem \ref{th:RecLatGeneralGroup}    for residually finite algebras immediately yields the following consequence, even though we also can give a direct proof without using residual finiteness.
\begin{theorem}\label{thm:mainZ1}
Function $f:\Z\to\Z$ is 
$+$-congruence preserving if and only if for every  
recognizable subset $L$ of $\Z$ the lattice $\+L_{\langle  \Z ;+\rangle}(L)$ is  closed under $f^{-1}$.
\end{theorem}

\begin{proof} For any $a\in\Z$, in the algebra $\+Z$, $\gen(a)=\{c+a\mid c\in\Z\}=\Z$. Hence for any function $f\colon\Z\to\Z$, for any $a$, condition $f(a)\in \gen(a)$ trivially holds.
Thus by Theorem \ref{th:RecLatGeneralGroup}, $f$ is $\+Z$-congruence preserving if and only if for every   $\+Z$-recognizable subset $L$ of $\Z$ the lattice $\+L_{\+Z}(L)$ is  closed under $f^{-1}$.
 \end{proof}
\begin{remark}\label{r:Nnepasse-pas-aZ}
1) The previous result shows that 
conditions $(1)_\N$ and  $(3)_\N$ of Theorem \ref{thm:ipl}  
can be extended when substituting $\langle  \Z ;+\rangle$ for $\langle  \N ;+\rangle$. 
However,  conditions $(2)_\N$  and  $(3)_\N$ of Theorem \ref{thm:ipl}  are no longer equivalent when substituting $\langle  \Z ;+\rangle$ for $\langle  \N ;+\rangle$ 
as shown by the counterexample exhibited  in 2).

2)  It is straightforward to see that, for any   $L$ finite,  $\+L_{\+Z}(L)$ is the set of {\em all } finite subsets of $\Z$.
 Consider $f:\Z\to\Z$ such that 
$f(k)=2^k$ for $k\in\N$ and $f(x)=x$ if $x<0$.
As $f^{-1}(a)$ is finite for every $a$,  the inverse image of any finite subset is a finite subset, and $\+L_{\+Z}(L)$ is closed under $f^{-1}$. However $f$ is  not congruence preserving: 
for instance, $2-0=2$ does not divide $f(2)-f(0)=2^2-2^0=3$.

3) Example \ref{ex:contrexRecNonRat} shows that in Theorem  \ref{thm:mainZ1} ``regular" cannot be  substituted  for ``recognizable".
\end{remark}

\begin{example} \label{ex:contrexRecNonRat}
Theorem \ref{thm:mainZ1}  does not hold if we substitute  ``regular" for ``recognizable". In {\rm \cite{benois}\/} it is shown that a regular  subset $L$ of $\Z$ is  of the form
$L=L^+\cup (-L^-)$ where $L^+,L^-$ are regular subsets of $\N$,
i.e., $L=-(d+S+d\N)\cup F\cup(d+R+d\N)$ with
$d\geq1$,
$R,S\subseteq\{x\mid0\leq x<d\}$, and
$F\subseteq\{x\mid-d<x<d\}$
(possibly empty).
 Consider the regular set
$L= 6+10 \N$;
  function $f$ defined by $f(x) = x^2$ is congruence preserving by Lemma \ref{lem:CPsurA}.
The set $f^{-1}(L) = (\{4,6\}+10\N)\cup -(\{4,6\}+10\N)$ 
contains infinitely many negative numbers.
Each set $L-t$ (for $t\in\Z$) contains only finitely many negative numbers
and the same is true for any finite union of finite intersections of $L-t$'s, 
and, in particular, for any set in $\+L_\+Z(L)$.
Thus, $f^{-1}(L)$ is not in $\+L_\+Z(L)$.
\end{example}

%
\section{Case of rings of $p$-adic  integers}
\label{ss:ZpZhat}
%
For rings $\Z_p$ of $p$-adic integers, 
the results are similar to those for the ring $\Z$. 

{\bf Representation of $p$-adic integers.} Let us recall some basic facts about $p$-adic integers. 
The set $\Z_p$ of $p$-adic integers  is the projective limit 
${\underleftarrow{\lim}} \langle\Z/p^n\Z;+,\times\rangle$
relative to the projections $\pi_{p^n\!,p^m}\colon \Z/p^n\Z\to\Z/p^m\Z$ 
 with $n\geq m$, such that $\pi_{p^n\!,p^m}(x)=x\pmod p^m$ for 
$x=0,\ldots,2^n-1$. 
 Every $p$-adic integer  can be represented as an infinite sum
 $\sum_{n=0}^\infty a_n p^n$ with $a_n\in\{0,\ldots,p-1\}$. 
 Addition is performed with carries as in the finite case (except that it goes from left to right).
 Elements of $\N$ are represented in $\Z_p$ 
 by sums with an infinite tail of $0$'s.
Elements of $\Z\setminus\N$ correspond in $\Z_p$ to  base $p$
representations with an infinite tail of digits all equal to $p-1$.
 For instance, (writing $a^\omega$ for an infinite tail of digits all equal to $a$)
we have $100110^\omega + (p-1)(p-1)(p-1)(p-2)(p-2)(p-1)^\omega=0^\omega$.
\subsection{About saturation and congruence preservation} 
\label{contrexPat}
We first show that $p$-adic integers  come in naturally for congruence preservation reasons, more precisely we study the extension of congruence preserving functions  $\N \to\Z$ into  congruence preserving function  $\Z \to A$.
 Example \ref{contrex:NZ}   below  shows  that, if we want to extend  all congruence preserving functions  $\N \to\Z$ into  congruence preserving functions  $\Z \to A$, the carrier set $A$ of the extended algebra cannot be reduced to $\Z$
and it must be ``saturated" in the sense of $p$-adic analysis (different from Definition \ref{df:satureCong}).
 
\begin{example} \label{contrex:NZ}{\rm
On $\langle\N;+\rangle$} we can define by induction a  congruence preserving function $F\colon\N\to\N$ such that $F(0)=0$, $F(1)=F(2)=2$ and for all $n>1$,
$F(2^n-1)\equiv 0\pmod {2^n}$.   
See the Appendix for the inductive proof.

We now show that $F$ cannot be extended into a congruence preserving function $F_{\Z}$ on $\Z$.
Indeed,  $F_{\Z}(-1)$  ought to satisfy the following  conditions
\begin{itemize}
\item $F_{\Z}(-1)\equiv F_{\Z}(2)=2\pmod 3$, hence $F_{\Z}(-1)\not=0$

\item $F_{\Z}(-1)\equiv 0\pmod {2^n}$ for all $n$ 
(if $F_{\Z}$ is congruence preserving then  $2^n$ divides 
$F_{\Z}(2^n-1) -F_{\Z}(-1)$, and it already divides $F_{\Z}(2^n-1)=F(2^n-1)$\;)\end{itemize}
hence $|F_{\Z}(-1)|\geq 2^{n}$ for all $n$ and that is impossible in $\Z$.
\end{example}

\noindent{\bf Saturation.} Let $f\colon\N\to\N$ be a congruence preserving function. In order to extend $f$ into a congruence preserving function in $-1$,  letting  $a=f(-1)$,  we must have: 2 divides $a - f(1)$, 3 divides $a - f(2)$, 4 divides $a - f(3)$, \dots Hence infinitely many conditions must hold. For every finite subset of this set of  conditions, there exists such an $a$ in $\Z$ (and also in $\N$) as proved in Example \ref{contrex:NZ}.  Unfortunately,  an $a$ satisfying {\em all} the conditions in the infinite set of conditions  does not  exist in $\N$, or $\Z$, because neither $\N$, nor $\Z$ are ``saturated"(in the logical sense). Saturated sets containing  $\Z$ are the sets of  $p$-adic integers.
\subsection{Residual finiteness of rings of $p$-adic integers}
\begin{definition}
 The {\em $p$-adic valuation} of $x\in\Z_p$ is the maximum $k$ such that $p^k$ divides $x$,
i.e., the number of heading zeroes  in the $p$-adic representation of $x$.
\end{definition}
\begin{remark}
The set $U$ of invertible elements of $\Z_p$ consists of all elements with null $p$-adic valuation,
i.e., those with $p$-adic representation $(a_n)_{n\in\N}$ such that $a_0\neq0$.
Thus, every element $x$ of $\Z_p$ can be written 
$x=p^n u$ where $n\in\N$ and $u\in U$.
\end{remark}
\begin{lemma} [ cf. Lenstra \cite{lenstra2}]\label{ZpPrincipal} 
$\Z_p$ is a principal ring, and all ideals of $\Z_p$ are of the form $\+I_n=p^n\Z_p$, with  $n\in\N$, or $\+I=\{0\}$
\end{lemma}

\begin{proof}
Let us recall the simple proof.
If $a$ is an element with minimum valuation in an ideal $I$ then $a=p^n u$ for some inversible $u$
and also $I\subseteq p^n\Z_p$.
Let $v$ be an inverse of $u$. Then $I$ contains $av\Z_p=p^nuv\Z_p=p^n\Z_p$.
\end{proof}
\begin{corollary}\label{congZp} The ring $\langle  \Z_p ;+,\times\rangle$ is residually finite.
\end{corollary}
\begin{proof}  By Lemma \ref{lem:CPsurA} and Lemma \ref{ZpPrincipal} a non trivial congruence on  $\Z_p$ is of the form
$$x\sim_n y \ {\text {iff}}\  x-y\in p^n\Z_p.$$
There thus are $p^n$ equivalence classes $a+ p^n\Z_p$ with $a\in \{0,\ldots,p^n-1\}$.

 Finally, the trivial congruence $x\sim y \ {\text {iff}}\  x=y$ is equal to the intersection of all non trivial congruences.
\end{proof}
Lemmata \ref{lem:CPsurA} and \ref{ZpPrincipal} imply the following
\begin{proposition} \label{th:CP-IDRsurZp} A function $f\colon \Z_p\rightarrow \Z_p$ is congruence preserving on $\langle  \Z_p ;+,\times\rangle$ if and only if it is satisfies the divisibility condition \eqref{eq:f freezing cp}.
\end{proposition}
 Let us recall a result proved in \cite{cgg15} (Theorem 24) for  functions satisfying condition \eqref{eq:f freezing cp}, hence by Corollary \ref{th:CP-IDRsurZp}  for congruence preserving functions.
\begin{proposition}
Every $\langle +,\times\rangle$-congruence preserving function $f:\N\to\Z$
extends to unique $\langle +,\times\rangle$-congruence preserving function
 $\widehat{f}:\Z_p\to\Z_p$.
\end{proposition}

In particular, this answers the problem raised in section \ref{contrexPat} 
and Example~\ref{contrex:NZ}.

\begin{corollary}
Every $\langle +,\times\rangle$-congruence preserving function $f:\N\to\Z$
extends to unique $\langle +,\times\rangle$-congruence preserving functions
$\widehat{f}_\Z:\Z\to\Z_p$.
\end{corollary}
However,  for $x\in\Z$,  $\widehat{f}_\Z(x)$  is not necessarily
in $\Z$  but in $\Z_p$, cf. Example \ref{contrex:NZ}.
 
 \subsection{Recognizability in $\langle  \Z_p ;+,\times\rangle$}\label{s:rec}
 We here give a simple characterization of recognizable subsets.
\begin{proposition}\label{prop:contZp}
Let $X$ be a subset of $\Z_p$.
The following conditions are equivalent:
\begin{enumerate}
\item[(i)]
$X$ is a recognizable subset of the ring $\langle  \Z_p ;+,\times\rangle$.
\item[(ii)]
$X$ is of the form $X=F+p^n\Z_p$ for some $n\in\N$ and some finite subset of
$\{0,\ldots,p^n-1\}$.
\end{enumerate}
\end{proposition}
\begin{proof}

$(i)\Rightarrow(ii)$.
Assume $(i)$ and let $\varphi:\Z_p\to M$ be a surjective $\{+,\times\}$-morphism
with $M$ finite. 
Since $\Z_p$ is a ring  so is $M=\varphi(\Z_p)$.
Let $K=\varphi^{-1}(0_M)$ 
be the kernel of $\varphi$.
The kernel $K$ of $\varphi$ is not reduced to $\{0\}$, otherwise $\varphi$ would be injective and $M$ infinite. 
Let $n$ be the smallest among the  $p$-adic  valuations  of the elements of $K$.
 Then $K\subseteq p^n \Z_p$. Let $z\in K$ have
 this smallest $p$-adic valuation: 
$z=p^n u$ with $u\in U$. We have for all $v\in\Z_p$, $p^n v=p^nu\times u^{-1}v=z\times u^{-1}v$, hence, as $K$ is an ideal and $z\in K$, we have $K\supseteq p^n\,\Z_p$, whence
 $K=p^n\,\Z_p$.
As $\varphi^{-1}(\varphi(x))=x+K$, we get
$\varphi^{-1}(\varphi(x))=x+p^n\Z_p$.
As $M$ is finite, any $T\subset M$ is a finite union $T=\cup_{i=1,\ldots,k}\{\varphi(x_i)\}$.
Any $\{+,\times\}$-recognizable subset $Z$ of $\Z_p$ is thus a finite union of such
$\varphi^{-1}(\varphi(x_i)),\ i=1,\ldots,k$, i.e.,  
$Z=\bigcup_{i=1}^{i=k}(x_i+p^{n}\Z_p)=F+p^n\Z_p$ where
$F=\{x_1,\ldots,x_k\}$.
This proves condition $(ii)$.
\smallskip\\
 $(ii)\Rightarrow(i)$.
It suffices to consider $M=\Z/p^n\Z$ and $\varphi$  the modular projection
$(a_k)_{k\in\N}\mapsto(a_0,\ldots,a_{n-1},0,0,\ldots)$.
\end{proof}
%

\subsection{Congruence preserving functions and lattices} \label{s:latticesZp}

%
 Theorem~\ref{th:RecLatGeneralGroup} allows 
to extend  Theorem \ref{thm:mainZ1}   to the ring $ \Z_p$.

\begin{proposition}\label{prop:CPlatZp}
 A function $f:\Z_p\to\Z_p$ is $\langle \Z_p;+,\times\rangle$-congruence preserving if and only if for every  recognizable subset $L$ of $\Z_p$, the lattice $\+L_{\langle \Z_p;+,\times\rangle}(L)$ is  closed under $f^{-1}$.
\end{proposition}
\begin{proof} By  Corollary \ref{congZp} and 
 Theorem~\ref{th:RecLatGeneralGroup},
 noting that the condition $f(a)\in \gen(a)$ holds
 because, $\Z_p$ being a group, for any $a\in\Z_p$ we have $\gen(a)=\Z_p$.
\end{proof}
\if34
\subsection{Additive groups of $p$-adic } 

{\color{red} On devrait peut etre supprimer tous ces profinis et Cie qui m'ont valu de me casser la figure pour mon expose ???.}

The set $\Zhat$ of profinite integers  is the projective limit 
${\underleftarrow{\lim}} \,\Z/n\Z$
relative to the projections $\pi_{n,m}\colon \Z/n\Z\to\Z/m\Z$ for {\sout{$n\geq m$. }} {\color{red}{$m$ divides $n$}, i.e., $\Zhat=\{(a_n)_{n=1}^\infty\mid a_n\in\Z/n\Z \text{  and } \forall n,m \text { such that }m\mid n,\ 
a_n\equiv a_m \pmod m\}$}.

Recall \cite{wiki} that any limit of finite groups is residually finite. Hence both $\langle  \Z_p ;+\rangle$ and $\langle\Zhat;+\rangle$ are residually finite. Then Theorem \ref {th:RecLatGeneralGroup} implies
\begin{proposition}\label{prop:groupZpZhat} Let $\+A=\langle A;+\rangle $, with $A$ being  either 
$\Zhat$ or $\Z_p$.  A function $f\colon A\to A$ is congruence preserving if and only if for every  $\+A$-recognizable subset $L$, the lattice $\+L_{\+A}(L)$ is  closed under $f^{-1}$.
\end{proposition}

\begin{proof} Easy: note that $\+A$ is residually finite and for any $a\in A$, $\gen(a)=A$.
\end{proof}

{\color{red} Estce juste ???=  ca a l'air mais il me semble qu on avait parle de topologie pour + .
Voir ci dessous.}

The  topology comes in for the additive structure. 
Recall that the rings $\Z_p$ 
contains the ring $\Z$ 
and is a compact topological ring
for the topology $\+O
$ given by the ultrametric $d$ such that
$d(x,y)=2^{-n}$ where $n$ is largest such that $p^n$ 
divides $x-y$, i.e., $x$ and $y$ have the same first $n$ digits
in their base $p$ representation.
We proved in \cite{cgg15} that on $\Z_p$   every congruence preserving function $f$ is continuous and is the inverse limit of an inverse system of congruence preserving functions $f_n$, e.g., on $\Z_p$  we have
$f={\underleftarrow{\lim}} f_n$ where each $f_n$ $f_{n}\colon \Z/{p^n}\Z\to\Z/{p^n}\Z$  is congruence preserving, and moreover the $f_n$ form an inverse system, i.e., for $j\leq i$ ,  $\pi_{i,j} \circ f_i = f_j\circ \pi_{i,j}$, i.e., the following diagram commutes:

\centerline {
$\begin{CD}\label{eqComm}
 \Z/{p^i}\Z  @> f_i>>   \Z/{p^i}\Z \\
@VV\pi_{i,j}V        @VV\pi_{i,j}V\\
\Z/{p^j}Z   @> f_j>>  \Z/{p^j}\Z
\end{CD}$
 }

 \begin{definition}[Continuous recognizability] \label{def:crec}
Let  $\Xi$ 
  be a signature. 
Let $\+A=\langle A; \Xi,\+O\rangle$ be an algebra with signature $\Xi$ and with a topology  $\+O$.
A  subset of $A$ is said to be {\em continuously $\+A$-recognizable} if it is $\Xi$-recognized by a continuous
$\Xi$-morphism $\varphi:A\to M$, where the finite algebra $M$ is endowed with
the discrete topology.
\end{definition}

\begin{proposition}\label{prop:contZp1}
Let $X$ be a subset of $\Z_p$.
The following conditions are equivalent:
\begin{enumerate}

\item[(i)]
$X$ is of the form $X=F+p^n\Z_p$ for some $n\in\N$ and some finite subset of
$\{0,\ldots,p^n-1\}$.

\item[(ii)]
$X$ is a continuously recognizable subset of the topological  group $\langle  \Z_p ;+,\+O\rangle$.
\end{enumerate}

\end{proposition}
\begin{proof}
$(ii)\Rightarrow(i)$.
Assume $(ii)$ and let $\varphi:\Z_p\to M$ be a continuous surjective $\{+\}$-morphism
with $M$ finite.
Then $M$ is an abelian group.
Let $K=\varphi^{-1}(0_M)$ (where $0_M=\varphi(0)$) be the kernel of $\varphi$.
As the topology on $M$ is  discrete,  any subset of $M$ is clopen; as  $\varphi$ is continuous, the inverse image $K$ of the clopen set $\{0_M\}$  is clopen in $\Z_p$, hence of the form $G_n+p^n\Z_p$, with $G_n\subseteq\{0,1,\ldots,p^n-1\}$ (by Proposition \ref{p:Zp compact} in the  Appendix). 
Following the same argument as above, we get condition $(i)$. 
\\
$(i)\Rightarrow(ii)$.
It suffices to consider $M=\Z/p^n\Z$ and $\varphi$  the modular projection
$(a_k)_{k\in\N}\mapsto(a_0,\ldots,a_{n-1},0,0,\ldots)$.
\end{proof}
\fi

\begin{remark}
We do not know whether  Proposition \ref {prop:CPlatZp}  extends to the ring of profinite integers $\Zhat$. As there is no simple characterization of the ideals of $\Zhat$  \cite{lenstra2}, 
 our proofs of Corollary \ref{th:CP-IDRsurZp} and Proposition \ref{prop:contZp}  do not hold for $\Zhat$.  We nevertheless conjecture that  these results  hold on $\Zhat$.
\end{remark}
\section{Conclusion}

We studied the relationships between lattices generated by recognizable sets of some  algebras  and  congruence preserving functions.  For quite a few usual algebras (in particular those with carriers $\N,\  \Z,\ \Z_p$), we showed that congruence preserving functions somehow correspond to functions 
 which can be added to the algebra with no modification of the lattices generated by 
recognizable set and closed under the inverses of the ``generating operations"  of the algebra.

{\small

}
\section{Appendix}

%
%
\begin{proof} [Proof of Example \ref{contrex:NZ}]
 We can define, by induction on $x\in\N$, a  congruence preserving function $F\colon\N\to\N$ such that:  $F(\N)\subseteq \N$, $F(0)=0$, $F(1)=F(2)=2$ and for all $n>1$,
$F(2^n-1)\equiv 0\pmod {2^n}$. Note that $F$ is congruence preserving both as a mapping $\N\to\N$, and as a mapping $\N\to\Z$.  
Basis: for $x=3$ ($n=2$), $F(3)=12$ is suitable

Induction: assume $F(y)$ has been defined for $y< x$ and define $F(x)$; $F(x)$ must satisfy
\begin{subequations}\label{system0}
\begin{align}
F(x)&\equiv F(0) \pmod {2^n} &  \text{ if } x= 2^n-1 \label{eq0bis}\\
F(x)&\equiv F(0) \pmod x  \label{eq0} \\
F(x)&\equiv F(1) \pmod {(x-1)}\notag\\
\notag\vdots\\
F(x)&\equiv F(i) \pmod {(x-i)} \\
\notag\vdots\\
F(x)&\equiv F(x-2) \pmod {2} \label{eqx-2}
\end{align}
\end{subequations}

\noindent Otherwise stated, assuming for all $n<x$, $F(n)$ satisfying \eqref{system0} has been defined, we define  $F(x)$ satisfying  all of the  equivalences in \eqref{system0}. 

As for $q_1,q_2$   coprime, $a\equiv b\pmod {q_1 \times q_2}$ if and only if 
$a\equiv b\pmod {q_1}$ and $a\equiv b\pmod  {q_2}$, we can transform the above system of equivalences \eqref{system0} so that all equivalences are modulo a power of a prime number. For any $n$ and any prime $p<n$, let $\alpha_{n,p}$ be the exponent of the largest power of $p$  dividing $n$ (i.e., $\alpha_{n,p}=\lfloor
\log_pn\rfloor$).
Then  system  \eqref{system0} is  equivalent to system \eqref{system1} below

\begin{subequations}\label{system1}
\begin{align}
F(x)&\equiv F(0) \pmod {2^n}  &\text{ if } x= 2^n-1\quad\text{ and } p=2  \\
F(x)&\equiv F(0) \pmod {p^{\alpha_{x,p}} }&\text{ for all $p$ dividing } x\\
F(x)&\equiv F(1) \pmod {p^{\alpha_{(x-1),p}} } \quad&\text{ for all $p$ dividing } x-1\\
\notag\vdots\\
F(x)&\equiv F(i) \pmod {p^{\alpha_{(x-i),p}}}\ &\text{ for all $p$ dividing } x-i\\
\notag\vdots\\
F(x)&\equiv F(x-2) \pmod {2}&\text{ for all $p$ dividing } 2 
\end{align}
\end{subequations}

The next Fact allows us to simplify system \eqref{system1} in such a way that each prime $p$ occurs in at most one equivalence.

\begin{fact}\label{lemmered}
Let $F$ satisfy system \eqref{system0}  for all $y<x$. Let $p \leq x$ be a prime.  Let $\alpha_p=\max\big\{\alpha_{i,p} \ |\  i=2,\ldots,x\big\}$. Let  \eqref{system2} be the subsystem of \eqref{system1} consisting only of  equivalences modulo a power of the chosen prime $p$
\begin{subequations}\label{system2}
\begin{align}
F(x)&\equiv F(0) \pmod {2^n}  \text{ if }\ x= 2^n-1\ \text{ and }\  p=2  \label{eq0bis2}\\
F(x)&\equiv F(0) \pmod {p^{\alpha_{x,p}}} \text{ if $p$ divides } x \label{eq02}\\
\notag\vdots\\
F(x)&\equiv F(j) \pmod {p^{\alpha_{(x-j),p}} } \text{ if $p$ divides } x-j \label{eqj2}\\
\notag\vdots\\
F(x)&\equiv F(i) \pmod {p^{\alpha_{(x-i),p}}}\text{ if $p$ divides } x-i \label{eqi2}\\
\notag\vdots\\
F(x)&\equiv F(x-2) \pmod {p^{\alpha_{2,p}}}\text{ if $p$ divides } 2 
 \label{eqn-22}
\end{align}
\end{subequations}
The system \eqref{system2} of equivalences reduces to a {\em  single equivalence} \eqref{eqred}  of the form

\begin{subequations}\label{eqred} 
\begin{align}
\text {either } \quad F(x)&\equiv
F(i) \pmod {p^{\alpha_p}}\quad\text{ if } x\not=2^n-1\text{ or } p\not=2 \label{eqred1}\\
\text {or }\ \quad F(x)&\equiv
F(0) \pmod {2^n} \quad\text{ if } x=2^n-1\text{ and } p=2 \label{eqred2}
 \end{align}
\end{subequations}
\end{fact}

\begin{proof}
Clearly, if all equivalences \eqref{system2} hold then equivalence \eqref{eqred} holds. Let us prove the converse:
\begin{itemize}
\item assume first $x\not=2^n-1$  or $p\not=2$, then equation \eqref{eq0bis2} is out of the picture and system \eqref{system2} is equivalent to equation \eqref{eqred1}. If $\alpha_p={\alpha_{(x-i),p}}={\alpha_{(x-j),p}}$, arbitrarily choose one of the equivalences \eqref{eqj2} or  \eqref{eqi2} for \eqref{eqred}.
 We prove that equivalences  \eqref{eqj2} hold for all $j\in\{2,\ldots,x\}\setminus\{i\}$. Let $j$ be given, note that
$p^{\alpha_{(x-j),p}} $ divides $p^{\alpha_{(x-i),p}} $
\begin{itemize}
\item  assume e.g. $i>j$. Then $p^{\alpha_{(x-j),p}} $ divides both $x-i$ and $x-j$, hence it divides $(x-j)-(x-i)=i-j$, and by the induction hypothesis, equivalences \eqref{system0} hold for $i$, thus $F(i)\equiv F(j) \pmod {p^{\alpha_{(x-j),p}} }$ 
\item  \eqref{eqred1}  and $p^{\alpha_{(x-j),p}} $ divides $p^{\alpha_{(x-i),p}} $
imply that   $F(x)\equiv F(i) \pmod {p^{\alpha_{(x-j),p}}}$
\end{itemize}

 hence \eqref{eqj2} by transitivity.
 
 \item assume now $x=2^n-1$. For $p\not=2$ the proof is the same as above. If    $p=2$, then equation \eqref{eq0bis2} gets into the picture and more care is needed. 
  The largest power of 2 dividing the ($x-i$)'s is $2^{n-1}$ and it is reached for $i=2^{n-1}-1$ as 
 $x-i=2^n-1- 2^{n-1}-1=2^{n-1}$. 
 \begin{itemize}
 \item By the induction hypothesis on $F$, we have $F(2^{n-1}-1)\equiv 0\pmod{2^{n-1}}$
 \item As $F(0)=0$, equation \eqref{eqred2} implies that $F(x)\equiv 0\pmod{2^{n}}$
 \end{itemize}
Hence $F(x)\equiv F(2^{n-1}-1)\pmod{2^{n-1}}$ by transitivity of $\equiv$  
and, by the proof of the previous case, this implies that all equations \eqref{eq02} to \eqref{eqn-22} hold.
 \qedhere
  \end{itemize}
\end{proof}

We now use the above  Fact  \ref{lemmered}  to simplify system \eqref{system1} so that there is at most one equivalence for each prime $p$; 
the Chinese remainder theorem shows that  this last system has at least a solution $F(x)=y_0$;
as $F$ is congruence preserving both $2^n-1$ and $2^n$ divide $F(2^n-1)$. 
\end{proof}

\subsection*{Profinite integers}
Recall some classical equivalent approaches to
the topological rings of $p$-adic integers  integers,
cf. Lenstra \cite{lenstra2} and Lang \cite{lang}.
\begin{proposition}\label{p:Zp}
Let $p$ be prime.
The three following approaches lead to isomorphic structures,
called the topological ring $\Z_p$ of $p$-adic integers.
\begin{itemize}
\item
The ring $\Z_p$ is the inverse limit of the following inverse system:
\begin{itemize}
\item
the family of rings $\Z/p^n\Z$ for $n\in\N$,
endowed with the discrete topology,
\item
the family of surjective morphisms
$\pi_{p^n,p^m}:\Z/p^n\Z \to \Z/p^m\Z$ for $0\leq n\geq m$.
\end{itemize}
\item
The ring $\Z_p$ is the set of infinite sequences $\{0,\ldots,p-1\}^\N$
endowed with the Cantor topology and addition and multiplication
which extend the usual way to perform addition and multiplication
on base $p$ representations of natural integers.
\item
The ring $\Z_p$ is the Cauchy completion of the metric topological ring
$(\N,+,\times)$
relative to the following ultrametric: 
$d(x,x)=0$ and for $x\neq y$, $d(x,y)=2^{-n}$ where $n$ is the $p$-valuation of $|x-y|$,
i.e., the maximum $k$ such that $p^k$ divides $x-y$.
\end{itemize}
\end{proposition}
%
\begin{proposition}\label{p:Zp compact}
The topological group $\Z_p$ is compact and has a basis of clopen sets (sets which are both open and closed) of the form $\prod_{i=0}^nG_i\times\prod_{i=n+1}^\infty\Z/p^i\Z$, with $G_i$ a subset of $\{0,\ldots,p-1\}$. 

Any clopen set of  $\Z_p$  is of the  form
$G_n+p^n\Z_p$, with $G_n\subseteq\{0,1,\ldots,p^n-1\}$. 
\end{proposition}
\begin{proof} The topology on $\Z_p$ is the  product topology (coarsest such that the projections $\pi_i\colon\Z_p\to\Z$ defined by $\pi_i(a_0,a_1,\ldots)=a_i$ are continuous); as $\Z/p^i\Z$ has the discrete topology, any $G_i\subseteq \Z/p^i\Z$ is both open and closed; hence any set $\pi_i^{-1}(G_i)=\prod_{i=0}^{n-1}\Z/p^i\Z\times G_i\times\prod_{i=n+1}^\infty\Z/p^i\Z$ is clopen, and so is also any finite intersection of such sets, $\prod_{i=0}^nG_i\times\prod_{i=n+1}^\infty\Z/p^i\Z$.
It is easy to see that these clopens form a basis of the open sets of $\+O$. 

$\Z_p$ is compact: by the theorem of Tychonoff the product of compact topological spaces is itself compact. $\Z_p$  is an intersection of closed subsets:  $\Z_p=\cap _{n=0}^\infty C_n$, where
$C_n=\prod_{i=0}^{n-2}\Z/p^i\Z\times G_n\times\prod_{i=n+1}^\infty\Z/p^i\Z$  where $G_n=\{(a_{n-1},a_n)\mid \pi_{n-1} (a_n)=a_{n-1}\}\subset \Z/p^{n-1}Z\times\Z/p^n\Z$; $G_n$ is closed as it is a finite subset of the product of two finite sets with discrete topology, the complement $\overline{G_n}$  of $G_n$ is thus open and so is the complement $\overline{C_n}$ of $C_n$; thus 
$C_n$ is closed in $\prod _{i=0}^\infty\Z/p^i\Z$ ;
$\Z_p$ is therefore closed as intersection of closed subset, hence compact.  

Finally, all clopens of $\Z_p$ are of the form $\prod_{i=0}^nG_i\times\prod_{i=n+1}^\infty\Z/p^i\Z$, with $G_i$ a subset of $\{0,\ldots,p-1\}$: any clopen set $L$ is a union $\cup_{i\in I} C_i$ of such clopens of the basis; $L$ being a closed subset of a compact space is compact, hence a finite union $\cup_{i\in \{i_,\ldots,i_k\}} C_i$ of these clopens covers $L$,  and as
$L=\cup_{i\in I} C_i$, we also have $L=\cup_{i\in \{i_,\ldots,i_k\} } C_i$. It is then easy to see that $L$ is in the form $G_n+p^n\Z_p$, with $G_n\subseteq\{0,1,\ldots,p^n-1\}$.
\end{proof}
\end{document}